\documentclass[jmp,amsmath,amssymb,reprint,onecolumn]{revtex4-2}

\usepackage{amsmath,amssymb,amsfonts,amsthm} 
\usepackage{color}
\usepackage{epsfig}
\usepackage{graphicx}

\newtheorem{theorem}{Theorem}[section]

\newtheorem{lemma}[theorem]{Lemma}
\newtheorem{proposition}[theorem]{Proposition}

\numberwithin{equation}{section}

\newcommand{\II}{{\mathbb I}}
\newcommand{\CC}{{\mathbb C}}
\newcommand{\RR}{{\mathbb R}}

\newcommand{\NN}{{\mathbb N}}
\newcommand{\TT}{{\mathbb T}}
\newcommand{\ZZ}{{\mathbb Z}}

\newcommand{\cB}{{\mathcal{B}}}

\newcommand{\cD}{{\mathcal{D}}}
\newcommand{\cF}{{\mathcal{F}}}

\newcommand{\bomega}{{\mbox{\boldmath $\omega$}}}

\newcommand{\bg}{{\mbox{\boldmath $g$}}}
\newcommand{\bp}{{\mbox{\boldmath $p$}}}
\newcommand{\sbp}{{\mbox{\footnotesize \boldmath $p$}}}
\newcommand{\bx}{{\mbox{\boldmath $x$}}}
\newcommand{\sbx}{{\mbox{\footnotesize \boldmath $x$}}}
\newcommand{\by}{{\mbox{\boldmath $y$}}}
\newcommand{\sby}{{\mbox{\footnotesize \boldmath $y$}}}

\newcommand{\bO}{{\mbox{\boldmath $O$}}}
\newcommand{\sbO}{{\mbox{\footnotesize \boldmath $O$}}}

\newcommand{\bP}{{\mbox{\boldmath $P$}}}
\newcommand{\sbP}{{\mbox{\footnotesize \boldmath $P$}}}
\newcommand{\bQ}{{\mbox{\boldmath $Q$}}}

\newcommand{\fA}{{\mathfrak A}}
\newcommand{\bfA}{{\mbox{\boldmath $\mathfrak A$}}}
\newcommand{\ofA}{\overline{\mathfrak A}}
\newcommand{\obfA}{\overline{\mbox{\boldmath $\mathfrak A$}}}

\newcommand{\fF}{{\mathfrak F}}
\newcommand{\ofF}{\overline{\mathfrak F}}

\newcommand{\fR}{{\mathfrak R}}

\newcommand{\bpartial}{{\mbox{\boldmath $\partial$}}}

\definecolor{olive}{RGB}{60,120,0}

\newcommand{\scirc}{\mbox{\footnotesize $\circ$}}

\newcommand{\ad}[1]{\mbox{Ad} \, #1}
\newcommand{\Imm}[1]{\mbox{Im} \, #1}
\newcommand{\Ree}[1]{\mbox{Re} \, #1}

\def\ie{{\it i.e.\ }}
\def\viz{{\it viz.\ }}
\def\etc{{\it etc}}

\newcommand{\be}{\begin{equation}}
\newcommand{\ee}{\end{equation}}  

\begin{document}

%\preprint{}

\title{\Large Trapped bosons, thermodynamic limit and condensation:  
a study in the \\[-1mm] framework of resolvent algebras} 



\author{Dorothea Bahns}
\email[]{dorothea.bahns@mathematik.uni-goettingen.de}
\author{Detlev Buchholz}
\email[]{detlev.buchholz@mathematik.uni-goettingen.de}
%\thanks{}
\affiliation{ Mathematisches Institut, Universit\"at G\"ottingen,
  B\"urgerstraße 40, 37073 G\"ottingen - Germany }


%\date{\today}

\begin{abstract} \noindent
The virtues of resolvent algebras, compared to other 
  approaches for the treatment of canonical quantum systems, are
  exemplified by infinite systems of non-relativistic   
  bosons. Within this framework, equilibrium states of trapped
  and untrapped bosons are defined on a fixed C*-algebra for
  all physically meaningful values of the temperature and chemical potential.
  Moreover, the algebra provides the tools for their analysis without
  having to rely on \textit{ad hoc}
  prescriptions for the test of pertinent features, such as the
  appearance of Bose-Einstein condensates. The method is illustrated in case of 
  non-interacting systems in any number of spatial dimensions
  and sheds new light on the appearance of condensates.
  Yet the framework also covers interactions and thus  
  provides a universal basis for the analysis of bosonic systems. 

  \medskip \noindent 
  Dedicated to Jakob Yngvason on the occasion of his 75th birthday
\end{abstract}

%\keywords{}

\maketitle

% \section{}

%
\vspace*{-4mm}
\section{Introduction}
\setcounter{equation}{0}

\noindent The algebraic approach to the discussion of infinite   
bosonic systems appears to be hampered by the fact that bosons
can accumulate unlimitedly in finite regions. It might therefore seem 
hopeless to find observables which still give meaningful
results in such singular states. Indeed, algebras of polynomials of the
underlying Bose fields do not work. Proceeding to the common Weyl 
algebra generated by their exponentials, the unitary Weyl operators, 
does not solve this problem either. These operators
can often not be interpreted as meaningful observables in such
singular states in view of their unrestrained fluctuations
and discontinuous
behavior under symmetry transformations. In applications one
encounters this problem by observing, for example, that the 
limits of certain sequences of equilibrium states on the Weyl algebra 
no longer satisfy the distinctive KMS condition.

\medskip
A surprisingly simple solution of these well known problems was    
presented in \cite{BuGr2}. Instead of dealing with
exponentials of the Bose fields, one considers their resolvents.
The resulting C*-algebra is, such as the Weyl algebra, defined
by a few relations, encoding the algebraic properties of the field. 
But in contrast to the Weyl algebra, this resolvent algebra
is not simple, \ie it has a non-trivial ideal structure \cite{Bu1}.
As a matter of fact, these ideals are a necessary ingredient for the
solution of the preceding problems, since they comprise those 
observables which become singular, hence physically meaningless in
singular states. These observables are
annihilated in the corresponding representations, \ie they are
members of some ideal of the resolvent
algebra and do not cause any problems in these states. 
On the other hand, the resolvent algebra has faithful irreducible
representations induced by states of physical interest \cite{BuGr2}.
It therefore covers in a meaningful manner all possible
states of bosonic systems. 

\medskip
It ought to be mentioned that resolvent algebras had already
appeared in disguise much earlier in an investigation by Kastler of
potentially interesting algebras for the description of canonical
quantum systems \cite{Ka}. This relationship was recently uncovered  
by Georgescu and Iftimovici \cite{GeIf}, who also added further
structural results to this framework. In spite of the 
work of Kastler, the utility of the resolvent algebra
for the discussion of infinite bosonic systems remained
unnoticed for decades, however. Only more recently
it has found applications
to problems of physical interest, cf.\ for example
\cite{Bu2,BuGr1,BuGr3,Co,KaMa}. 

\medskip
Within this framework, we complement here these results by an    
analysis of trapped and untrapped equilibrium states of non-interacting bosons
for given temperature and chemical potential.
Trapped states are 
described by Gibbs-von Neumann density operators on Fock space
and are called Gibbs-von Neumann states, for short.
We study in detail their thermodynamic and infinite particle
number limits, the latter being obtained by
proceeding to the maximal admissible value of the chemical potential,
and the appearance of Bose-Einstein condensates. These topics have
been widely discussed in the literature from various points of view,
cf.\ \cite{BrRo, LiSeSoYn, Ve} and references quoted there. 
It is the primary purpose of our article to illustrate the  
virtues of the resolvent algebras, compared to these other
approaches, and to shed new light on some well 
known results. For the convenience of the reader, we begin
by briefly recalling from \cite{BuGr2} some relevant definitions
and pertinent facts. 

\medskip
As already mentioned, the resolvent algebra can abstractly be defined
by relations. Since it is faithfully represented on Fock space
\cite{BuGr2}, we can deal here with this concrete representation
and proceed from it later to disjoint representations of physical
interest. The bosonic Fock space is denoted by $\cF$;
it is the infinite direct sum $\bigoplus_{n = 0}^\infty \, \cF_n$ of 
$n$-particle spaces which are generated by $n$-fold 
symmetric tensor products of single particle states in 
$\cF_1 \simeq L^2(\RR^s)$ in $s$ spatial dimensions. 
The zero particle state in $\cF$ (vacuum) is denoted by $\Omega$. 

\medskip
Let $\cD(\RR^s) \subset L^2(\RR^s)$
be some space of complex-valued test functions
on configuration space. As a matter of fact,
one may proceed from any dense subspace of $L^2(\RR^s)$,
our results do not depend on its particular choice.
We will comment on this point further below.
The Bose field underlying Fock space is denoted
by $\phi$. It is defined as a real linear map from $\cD(\RR^s)$ to selfadjoint
operators, acting on a common core in $\cF$, \viz the domain
of the particle number operator $N$. The corresponding
annihilation and creation operators are given by 
$a(f) = 2^{-1/2} \big(\phi(f) + i \phi(i f) \big)$
and
$a(f)^* =  2^{-1/2} \big( \phi(f) - i \phi(i f) \big)$,
$f \in \cD(\RR^s)$.
In order to avoid the subtleties associated with unbounded operators,
it is appropriate to proceed to bounded functions of the field.
Continuous functions of the field, vanishing at infinity, are
most convenient and prominent examples are the  resolvents 
\be
R(\lambda, f) \doteq (i \lambda 1 - \phi(f))^{-1} \, ,
\quad f \in \cD(\RR^s) \, , \ \Ree{\lambda} \neq 0  \, .
\ee
Moreover, these
resolvents comprise all algebraic properties of the field.
For example, the canonical commutation relations are encoded in
the relation 
\be
[R(\lambda, f), R(\mu,g)] = i \, \Imm{\langle f, g \rangle} \
R(\lambda, f)  R(\mu,g)^2 \, R(\lambda, f) \, ,
\ee
where $\langle f, g \rangle$ denotes the (sesquilinear) 
scalar product of $f,g$ in  $L^2(\RR^s)$, fixing the
symplectic form $\Imm{\langle f, g \rangle} $. 
The resolvent algebra $\fR$ is defined as the 
norm-closed subalgebra
of the algebra $\cB(\cF)$ of bounded operators on $\cF$, which is
generated by the resolvents; it is thus a unital
C*-algebra  \cite{BuGr2}. 

\medskip
In the discussion of infinite bosonic systems, occupation numbers
of particle states
play a prominent role. The corresponding observables can be identified
within $\fR$ with the help of
the gauge group generated by the particle number operator $N$.
It operates by automorphisms $\gamma$ on $\fR$ and acts on 
the basic resolvents according to 
\be
\gamma(u)\big(R(\lambda,f)\big) \doteq e^{iuN} R(\lambda, f) \, e^{-iuN}
= R(\lambda, e^{iu} f) \, , \quad 0 \leq u \leq 2 \pi \, .
\ee
Even though this action is discontinuous in $u$ with regard 
to the norm topology, harmonic analysis of 
the gauge group is possible in
the C*-algebra $\fR$ in the following sense. One finds \cite{Bu4}
that for any $R \in \fR$ the integrals
\be  \label{e.1.4}
R_m \doteq (1/2 \pi) \int_0^{2 \pi} \! du \, e^{imu} \gamma(u)(R) \, , \quad
m \in \ZZ \, ,
\ee
being defined on $\cF$ in the strong operator topology, are  
elements of the resolvent algebra $\fR$. In particular, the
operators corresponding to
$m = 0$ constitute a norm
closed subalgebra $\fA \subset \fR$,
which contains the gauge invariant observables.
The operators corresponding to arbitrary $m \in \ZZ$ transform
as tensors under the action of the gauge group and generate
a norm closed subalgebra $\fF \subset \fR$ on which this group acts
pointwise norm continuously. Thus, harmonic analysis and
synthesis of the gauge group is possible on $\fF$. It is
therefore convenient to focus on these subalgebras in applications. 

\medskip
Let us mention as an aside that the Weyl algebra, although being     
stable under the action of the gauge group, does not contain
a single gauge invariant observable, apart from multiples of the
identity. The integrals of Weyl operators,
analogous to relation \eqref{e.1.4}, 
are not contained in the Weyl algebra and therefore cannot be used in 
the analysis of non-Fock states on it. 

\medskip
The primary motivation for the introduction of resolvent algebras    
in \cite{BuGr1} was the observation that
they admit the automorphic action of a multitude of
dynamics of physical interest, in contrast to the
Weyl algebra. For bosonic systems with a finite number of
degrees of freedom, described by finite dimensional
test function spaces, the preceding framework is fully satisfactory.
Yet for infinite systems it is
in general too restrictive. This can
already be seen in case of non-interacting systems: let $H$ be
the second quantization of any given selfadjoint operator $h$ on the
single particle space $\cF_1$. Thus $H$ commutes with the particle
number operator $N$ and one may therefore ask whether the
algebra of observables $\fA$ is stable under the adjoint action
of $e^{itH}$, $t \in \RR$. It turns out that for any given $n \in \NN$
one has 
\be \label{e.1.5}
\ad{e^{itH} (A)}  \upharpoonright
   {\textstyle \bigoplus}_{m = 0}^n \, \cF_m \, \in \,
   \fA \upharpoonright
   {\textstyle \bigoplus}_{m = 0}^n \, \cF_m \, , \qquad
     t \in \RR \, , \ A \in \fA \, , 
\ee
\ie the restrictions of the observable algebra to states
with limited particle numbers are stable under
all of these dynamics.
Yet the hoped for inclusion $\ad{e^{itH}} (\fA) \subset \fA$ 
holds only if the chosen test function space is stable
under the action of the single particle dynamics,
\ie $e^{ith} \, \cD(\RR^s) \subset \cD(\RR^s)$, $t \in \RR$. 
As an aside, relation \eqref{e.1.5} still holds
in case of dynamics involving sufficiently regular pair interactions,
but the desired inclusion never applies. 

\medskip 
These problems can be solved by slightly extending the
observable algebra. Since this algebra works well for finite
systems, it seems natural 
to try to find an extension without
altering the observables on the subspaces
of Fock space with limited particle number. Adopting this 
point of view, the maximal admissible extension is the C*-algebra
$\ofA \supset \fA$ which is generated by all bounded, gauge invariant
operators on $\cF$ whose action on any given subspace 
${\textstyle \bigoplus}_{m = 0}^n \, \cF_m$ coincides with
the action of some operator in $\fA$, possibly depending
on $n \in \NN$. Alternatively, $\ofA$ can be presented as 
projective limit of the observable algebras on the
subspaces of $\cF$ with limited particle number.
We will refer to $\ofA$ as \textit{canonical extension} of $\fA$.
It is only a small subalgebra of
the algebra of all bounded, gauge invariant operators on $\cF$;
for example, it does not contain any operator with continuous
spectrum. Nevertheless, this algebra is stable under the action of
non-interacting dynamics as well as interacting dynamics
with continuous two-body potentials 
vanishing at infinity. In a similar manner, the
field algebra $\fF$, generated by all tensors under the action
of the gauge group, can be extended to a C*-algebra
$\ofF$ which is stable under these dynamics. As a matter of 
fact, this extension is obtained by adding to $\ofA$
a single isometry. For a thorough discussion of these
facts, cf.\ \cite{Bu3,Bu4}.

\medskip
The algebras $\ofA$, respectively $\ofF$, provide in a sense maximal 
arenas for the discussion of the dynamics of infinite bosonic
systems. Yet in applications it is frequently more convenient
to proceed to suitable subalgebras. For given dynamics, a minimal
choice for the
observables would be the C*-algebra, which is generated on $\cF$ by the 
operators $\{ \ad{e^{itH}} (A) : A \in \fA \, , \ t \in \RR \}$.
If one needs to have control on the continuity properties
of the dynamics, pointwise in the norm topology (characterizing
C*-dynamical systems), one is led to
consider the C*-algebra $\ofA_c$, generated by
\be \label{e.1.6} 
\Big\{ \int \! dt \, f(t) \, \ad{e^{itH}} (A) : A \in \fA \, ,
\ f \in L^1(\RR) \Big\} \, ,
\ee
where the integrals are defined in the strong operator topology.
It was shown in \cite{Bu3} that the algebra $\obfA_c$
is also contained in $\ofA$. 
In a similar manner one can proceed for given dynamics to convenient 
subalgebras of $\ofF$, cf.\ \cite{Bu4}.

\medskip
In the present article we restrict our attention to    
systems of non-interacting bosons. Since we are interested in trapped
systems and their thermodynamic limit, we need to consider different
dynamics. An algebra which covers them all is obtained by
choosing as test function space underlying the
resolvent algebra $\fR$ the space of all square integrable
functions $L^2(\RR^s)$. We will deal in the following with
this algebra and the corresponding observable and field
subalgebras~$\fA$, respectively $\fF$. It is note-worthy that
the latter two algebras are contained in the
canonically extended  algebras obtained 
for any initial choice of a test function space $\cD(\RR^s)$.
We also note that in the present case of non-interacting systems
most computations can conveniently be performed within
the resolvent algebra $\fR$,
because it is stable under the dynamics. But in case of
interaction it is advantageous to proceed from the outset
to its observable, respectively field subalgebras, where one has
better control of the action of the dynamics.

\medskip
Within this framework we consider Gibbs-von Neumann states,    
\ie trapped thermal equilibrium states of bosons for given temperature
and chemical potential, which are
faithful and normal with regard to the Fock representation.
By proceeding to the limits 
of zero trapping potential (thermodynamic limit), respectively 
infinite particle number (maximal chemical potential), 
we obtain states which still satisfy the
KMS-condition for the limit dynamics.
The resulting representations are not normal with
regard to the Fock representation, however. Moreover,
in the infinite particle number limit they 
are also not faithful, 
since observables, which are sensitive to
particles in the respective ground states become
singular, \ie the corresponding operators in
the resolvent algebra disappear in the limit. 
This is an indication for the appearance of
condensates. 

\medskip
The study of condensates requires a more detailed   
analysis, however. To this end we adopt a
local point of view: let $\bO \subset \RR^s$ be any
bounded region and let $N(\bO)$ be the corresponding
number operator on Fock space $\cF$, counting the
number of particles in $\bO$. Thus for any given \textit{normal} state
$\omega$, described by a density matrix on Fock
space, the quantity $\omega(N(\bO))$ is the
expected number of particles in $\bO$. 
For normal thermal states, being of interest here, this quantity is
finite. Now given any sequence $\omega_\nu$ of normal states,
$\nu \in \NN$, these expectation values may (a) either stay finite or
(b) approach infinity in the limit. In case~(a), the
limit states have a finite particle density in $\bO$; 
in order to increase it one needs to add to them
further particles. 
In case (b), the appearance of an infinity of particles
in~$\bO$ can have different reasons.
We attribute it here to the formation of
a Bose-Einstein condensate whenever there is some wave
function $g$ such that the number of particles in $\bO$
having wave functions in the orthogonal complement of~$g$
stays finite in the limit. So these particles 
have a critical density.
The divergence of $N(\bO)$
can then be attributed entirely to particles with
wave function $g$, creating condensates in the
approximating states which grow unboundedly. 

\medskip
The computations needed in the analysis of     
equilibrium states in order 
to  decide which of these cases is at
hand can conveniently be performed with aid of the 
resolvent algebra. It turns out that for trapped thermal systems, 
which are confined by some regular potential, case (b)
always occurs in the limit of infinite particle numbers.
Wheras the density of particles in excited states,
forming a \textit{thermal cloud}, remains bounded
in the approximating states, these states exhibit 
growing densities of Bose-Einstein condensates of particles
with wave function $g$ carrying minimal energy. This feature
occurs in any number of dimensions and does not imply the 
spontaneous breakdown of gauge invariance.

\medskip
In case of the thermodynamic  
limit states, containing an infinity of particles from the
outset, the local properties of the states depend on the number of
dimensions: in \mbox{$s = 1,2$} dimensions, the states obtained for
maximal chemical potential also belong to case~(b) without
including proper condensates. Yet, as we shall see, 
they have properties which resemble quasi-condensates.  
If $s > 2$ one arrives in the limit at case~(a), \ie the particle   
density in bounded regions stays finite. In order to increase this
density, one can modify the limit states without destroying
their equilibrium property. This is accomplished by shifting
the quantum field by a classical field which is 
created by the collective effects of particles with
a suitable (improper) 
wave function $g$. It can be described by the adjoint action of
Weyl operators, depending on~$g$. 
In this manner, the local density of particles with this
wave function can be arbitrarily increased without affecting the
density of the thermal cloud, escorting them.
Hence the resulting states describe condensates as well. 
These states are no longer gauge invariant, however,
due to the action of the
Weyl operators. So, in summary, the resolvent algebra provides
an adequate framework for the study of trapped and untrapped
thermal states and of the differing manifestations 
of condensation. 

\medskip
Our article is organized as follows. In the subsequent section
we consider quasifree states on the resolvent algebra which
are fixed by specifying their one and two-point functions.
For appropriate sequences of such states, we will study 
their convergence to limit states which are disjoint from the
Fock representation. The insights gained in this section
will be used in Sect.\ 3 in the analysis of trapped equilibrium
states and of their thermodynamic and infinite particle
number limits. The particle
density is recovered from the observable algebra in Sect.\ 4 and 
the formation and analysis of condensates 
is discussed in Sec.\ 5. The article
concludes with a summary and an outlook on interacting systems.

\section{Quasifree states}
\setcounter{equation}{0}

In this section we consider certain specific quasifree states on
the resolvent algebra which are normal with regard to the Fock
representation. We will proceed
from them to limit states, inducing inequivalent representations
by the GNS construction~\cite{Ha}. In this analysis
we make use of the fact that the Fock representation of the
resolvent algebra extends to a regular representation of the
Weyl algebra and \textit{vice versa}. Regular representations
of resolvent algebras are defined by the property that all
resolvents have trivial kernels. They are in one-to-one
correspondence to regular representations of the Weyl
algebra, in contrast to the singular
representations \cite[Cor.\ 4.4]{BuGr2}.

\medskip
The unitary Weyl operators in the Fock representation are denoted 
by $W(f)$, $f \in L^2(\RR^s)$. They satisfy the common Weyl relations 
\be 
W(f) W(g) = e^{- (i/2) \, \text{Im} \langle f, g \rangle} \, W(f + g) \, ,
\quad f,g \in L^2(\RR^s) \, .
\ee
The quasifree states $\omega$ of interest
here are conveniently defined on the Weyl operators
by the equality
\be \label{e.2.2}
\omega(W(f)) = e^{i \, l_\omega(f)} e^{- (1/2) \langle f, f \rangle_\omega } \, ,
\quad f \in L^2(\RR^s) \, .
\ee
Here $\langle f , g \rangle_\omega = \overline{\langle g , f \rangle}{}_\omega$
is a bilinear form on $L^2(\RR^s)$, regarded as a real vector space.
Its real part is a scalar product, its imaginary part is given by 
$(1/2) \, \Imm{\langle f, g \rangle}$, where   
$(1/4) \, |\Imm{\langle f, g \rangle}|^2 \leq \langle f, f \rangle_\omega
\,  \langle g, g \rangle_\omega$ for    
$f,g \in L^2(\RR^s)$.
To simplify terminology, we refer to the forms
$\langle  \, \cdot \,  , \, \cdot \, \rangle_\omega$ as 
scalar products. We assume in the following that they are
gauge invariant (stable under multiplication of its entries
with equal phase factors). The symbol~$\, l_\omega(\, \cdot \, )$
denotes a real linear functional on $L^2(\RR^s)$, which is
fixed by the one-point function of the field. The two-point
function of the field is related to these entities by the formula
\be \label{e.2.3}
\omega(\phi(f) \phi(g)) = \langle f, g \rangle_\omega +
l_\omega(f) l_\omega(g) \, , \quad f,g \in L^2(\RR^s) \, .
\ee
With the preceding conventions, a  quasifree state is gauge invariant if
and only if the linear functional vanishes.

\medskip
We are interested here in certain specific
sequences of such quasifree states and their
limits. As is well known, any sequence of states 
has limit points. But the resulting expectation values of
the Weyl operators have in general the form \eqref{e.2.2} only on
subspaces $\cD \subset L^2(\RR^s)$ and vanish on their complements. 
Moreover, the underlying linear functionals may attain non-linear
limits. Since the Weyl algebra is simple, such  states  
lead to singular GNS representations, often defying a 
meaningful physical interpretation. In particular,
it may happen that generators of the spacetime symmetries
are not defined. 

\medskip
As we shall see, the situation is better for
quasifree states on the resolvent algebra. 
The expectation values of the resolvents in 
these state are obtained from equation \eqref{e.2.2} by a Laplace 
transformation. For $\Ree{\lambda_1} > 0, \dots , \Ree{\lambda_n} > 0$
one has, cf.\ \cite[Eq.\ 17]{BuGr2},   
\begin{align}  \label{e.2.4} 
  & i^n  \omega\big(R(\lambda_1,f_1) \cdots R(\lambda_n, f_n)\big)  \\
  & = \int_{\RR_+^n} \! \! \! du_1 \cdots du_n \, 
  e^{- \sum_j u_j \lambda_j}
  \, \omega \big(W(- u_1 f_1) \cdots W(- u_n f_n)\big) \nonumber \\
  & = \int_{\RR_+^n} \! \! \! du_1 \cdots du_n \, 
  e^{- \sum_j u_j( \lambda_j + i l_\omega(f_j))} \, e^{ - (i/2)  
    \sum_{k < l} {\mathrm{Im}}
    \langle u_k f_k, u_l f_l \rangle} \, 
    e^{- (1/2)  \sum_{k,l} \langle u_k f_k , u_l f_l \rangle_\omega } \nonumber \, .
\end{align}  
Making use of the equality $R(-\lambda,f) = -R(\lambda,-f)$
one obtains analogous relations for any choice of the
signs of $\Ree{\lambda_1}, \dots , \Ree{\lambda_n}$. 

\medskip
Let us turn now to the analysis of sequences of quasifree states $\omega_\nu$ 
on the resolvent algebra,
which are fixed  by specific sequences of scalar products
and linear functionals. Having the applications in mind, we assume 
that the underlying sequences of
scalar products $\langle f , f \rangle_\nu$, $f \in L^2(\RR^s)$,
have (possibly infinite) limits as $\nu \in \NN$ tends to infinity.
With this input the following proposition obtains.

\begin{proposition} \label{p.2.1}
  Let $\omega_\nu$, $\nu \in \NN$, be a sequence of quasifree states on the
  resolvent algebra $\fR$ with properties specified above.
    \begin{itemize}
    \item[(i)] If the states are gauge invariant, the sequence    
      converges, pointwise on $\fR$, to some state $\omega_\infty$. 
      There is some complex subspace $\cD \subset L^2(\RR^s)$
      with scalar product
      $\langle \, \cdot \, , \, \cdot \, \rangle_\infty$, 
      such that for $f_1, \dots , f_n \in \cD$  and
      $\Ree{\lambda_1} > 0, \dots , \Ree{\lambda_n} > 0$,
      $n \in \NN$, 
      \begin{align}
        &  i^n \omega_\infty\big(R(\lambda_1, f_1) \cdots R(\lambda_n,f_n) 
        \big) \\
       & \! \! = \! \! \int_{\RR_+^n} \! \! \! \! du_1 \!  \cdots   du_n \, 
              e^{- \sum_j u_j \lambda_j} \, e^{ - (i/2)  
       \sum_{k < l} {\mathrm{Im}}
       \langle u_k f_k, u_l f_l \rangle} \, 
       e^{- (1/2)  \sum_{k,l} \langle u_k f_k , u_l f_l \rangle_\infty } \nonumber \, .
      \end{align}
      If 
      $\, \{f_1, \dots f_n\} \bigcap \big( L^2(\RR^s) \backslash \cD \big) 
      \neq \emptyset \,$ one has
      \be
      \omega_\infty\big(R(\lambda_1, f_1) \cdots R(\lambda_n,f_n)\big) = 0
      \, .
      \ee
    \item[(ii)] If the states are not gauge invariant, there exists    
      a state $\omega_\infty$ on $\fR$, which is a weak-* limit point
      of the sequence. It has the following properties:
      as in (i), there is some
      complex subspace $\cD \subset L^2(\RR^s)$ with a scalar product
      $\langle \, \cdot \, , \, \cdot \, \rangle_\infty$ and 
      a real subspace $\cD_\infty \subset \cD$ with a real
      linear functional $l_\infty$ such that for 
      $f_1, \dots , f_n \in \cD_\infty$  and
      $\Ree{\lambda_1} > 0, \dots , \Ree{\lambda_n} > 0$,
      $n \in \NN$, 
      \begin{align} \label{e.2.7}
      &  i^n \omega_\infty\big(R(\lambda_1, f_1) \cdots R(\lambda_n,f_n)\big) \\
      & \! \! = \! \! \int_{\RR_+^n} \! \! \! \! du_1 \!  \cdots   du_n \, 
        e^{- \sum_j u_j( \lambda_j + i l_\infty(f_j))} 
        e^{- (i/2) \sum_{k < l} {\mathrm Im}
    \langle u_k f_k, u_l f_l \rangle} 
    e^{- (1/2)  \sum_{k,l} \langle u_k f_k , u_l f_l \rangle_\infty} \nonumber .
      \end{align}
     If 
      $\, \{f_1, \dots f_n\} \bigcap \big( L^2(\RR^s) \backslash \cD_\infty \big) 
      \neq \emptyset \,$ one has
      \be
      \omega_\infty\big(R(\lambda_1, f_1) \cdots R(\lambda_n,f_n)\big) = 0
      \, .
      \ee
    \item[(iii)] Given any finite dimensional complex
      subspace $K \subset L^2(\RR^s)$,
      let \mbox{$\fR(K) \subset \fR$} be the subalgebra which
      is generated by resolvents $R(\lambda,f)$ with
      $f \in K$, $\Ree{\lambda} \neq 0$. There exists a subset
      $\II \subset \NN$ such that the sequence
      $\omega_\iota \upharpoonright
      \fR(K)$, $\iota \in \II$, converges pointwise 
      to $\omega_\infty$.
     \end{itemize}  
\end{proposition}
\noindent \textbf{Remark:} Resolvents assigned in (i) to test functions in 
$L^2(\RR^s) \backslash \cD$, respectively in (ii) to test functions in 
$L^2(\RR^s) \backslash \cD_\infty$, generate 
ideals in $\fR$. Note that these sets are not fixed
from the outset and depend on the limit state.

\begin{proof}
  Let $\cD \subset L^2(\RR^s)$ be the subset of all functions for which
  the sequences $\langle f, f \rangle_\nu$, $f \in \cD$, have a finite limit
  for $\nu$ tending to infinity. It follows from the triangle inequality
  and the polarization identity that $\cD$ is a complex linear
  subspace of $L^2(\RR^s)$ and that there is some scalar product
  $\langle \, \cdot \, , \, \cdot \, \rangle_\infty$
  on $\cD$ such that
  \be
      \lim_{l \rightarrow \infty} \langle f , g  \rangle_\nu =
      \langle f , g \rangle_\infty \, , \quad f,g \in \cD \, .
  \ee
  Now the function appearing in the
  exponents of the last line of equation \eqref{e.2.4},
  \begin{align} \label{e.kms}
    u_1, \dots, u_n & \mapsto
    (i/2) \sum_{k < l} \, \Imm{\langle u_k f_k, u_l f_l \rangle} +
    (1/2) \, \sum_{k,l} \, \langle u_k f_k, u_l f_l \rangle_\nu  \nonumber \\
    & =    \sum_{k < l} \, u_k u_l \, \langle f_k, f_l \rangle_\nu
    + (1/2)  \sum_k \, u_k^2 \, \langle f_k, f_k \rangle_\nu \, ,
  \end{align}
  has a non-negative real part. If $f_1, \dots , f_n \in \cD$, it 
  converges in the limit of large $\nu$ uniformly on compact subsets of
  $\RR^n$ to a similar expression, where the sequence of
  scalar products is replaced by their limits. On the other hand,
  if $f_k \in L^2(\RR^s) \backslash \cD$ for some $k$, the real part
  of this function tends to $+ \infty$ for almost all $u_1, \dots , u_n$.
  Now the modulus of the integrand in equation \eqref{e.2.4}
  is bounded by the integrable function
  $u_1, \dots , u_n \mapsto e^{- \sum_j u_j {\mathrm Re} \lambda_j}$.
  It therefore follows from the dominated convergence theorem that the
  expectation values $\omega_\nu(R(\lambda_1,f_1) \cdots R(\lambda_n,f_n))$
  in gauge invariant states converge in the limit of large $\nu$ 
  to the expressions given in part~(i) of the statement. Since
  the finite sums of products of resolvents are norm dense in the 
  resolvent algebra, it is then also clear that the sequence of
  gauge invariant states
  converges pointwise on $\fR$, completing the proof of part (i)
  of the statement.

  \medskip
  Turning to part (ii), the complex subspace
  $\cD \subset L^2(\RR^s)$ and the limit scalar
  product $\langle \, \cdot \, , \, \cdot \, \rangle_\infty$ are
  defined as in the preceding step. It also follows from that
  argument that the expectation values
  $\omega_\nu(R(\lambda_1,f_1) \cdots R(\lambda_n,f_n))$
  converge to $0$ if $f_j \in L^2(\RR^s) \backslash \cD$ for some
  $j$. Moreover, if $f_1, \dots , f_n \in \cD$, one can replace
  in equation \eqref{e.2.4} the 
  scalar products $\langle \, \cdot \, , \, \cdot \, \rangle_\nu$ 
  by their limit values
  $\langle \, \cdot \, , \, \cdot \, \rangle_\infty$,
  making use of the dominated convergence theorem.
  The resulting error vanishes in the limit of large~$\nu$. So we can 
  focus on the sequence of real linear functionals
  $l_\nu$, $\nu \in \NN$, restricted to~$\cD$.

\medskip
In order to exhibit the properties of this sequence, we proceed 
to the two-point compactification of $\RR$, given by the
function $x \mapsto c(x) \doteq x/\sqrt{x^2 + 1}$, where  
$\pm \infty$ are
mapped to $\pm 1$. We then consider the sequence of
functionals on $\cD$ given by $f \mapsto c(l_\nu(f)) \in [-1,1]$,
$\nu \in \NN$. According to Tychonoff's theorem the
cartesian product $\Pi_{f \in \cD} \, [-1,1]^f$, equipped with
the product topology, is compact. Hence the sequences
$\nu \mapsto c(l_\nu(f))$ have some limit (accumulation) point
in this product, which is 
denoted by $c(l_\infty(f))$, $f \in \cD$. Thus,
given any finite family of functions $f_j \in \cD$, 
there exists some subset $\II \subset \NN$ such that
$\lim_{\, \iota \in \II} c(l_\iota(f_j)) = c(l_\infty(f_j))$,
$j = 1, \dots , n$.

\medskip 
Now let $\cD_\infty \subset \cD$
be the set of all $f \in \cD$ for which 
$|c(l_\infty(f))| < 1$. Applying to
the corresponding approximating sequences the inverse function
$c^{-1}$, it follows that for any finite number of elements
$f_j \in \cD_\infty$ one has 
\mbox{$\lim_{\, \iota \in \II} l_\iota(f_j) =  l_\infty(f_j)$},
$j = 1, \dots , n$, for the appropriate choice
of index set $\II \subset \NN$. Since all functionals $l_\iota$ are
real linear on $\cD$, this implies that $\cD_\infty$ is a
real subspace of $\cD$ and that the limit $l_\infty$ is
a real linear functional on this space. Thus the
sequences $l_\iota(f)$, $\iota \in \II$, converge to $l_\infty(f)$ 
for all $f$ in the real linear span of the functions
$f_1, \dots , f_n$.

\medskip
Next, let $f_1 , \dots , f_n \in \cD$ be a finite   
family of functions among which there is
some {$f_k \in \cD \backslash \cD_\infty$}, \ie
$c(l_\infty(f_k)) \in \{ \pm 1 \}$.
As before, there is some subset $\II \subset \NN$
such that $\lim_{\, \iota \in \II} c(l_\iota(f_j)) = c(l_\infty(f_j))$,
$j = 1, \dots , n$. Applying again the inverse function $c^{-1}$, 
it follows that the sequence $l_\iota(f_k)$, $\iota \in \II$,
diverges. As a matter of fact, since  $f_1 , \dots , f_n$
span a finite dimensional space, the sequence of functions 
\be
u_1, \dots , u_n \mapsto \sum_{j = 1}^n u_j \, l_\iota(f_j) \, ,
\quad \iota \in \II \, ,
\ee
tends to $\pm \infty$ for almost all $u_1 , \dots , u_n \in \RR^n$. 

\medskip
With this information, we can turn now to the analysis of the expectation
values of products of
resolvents with functions \mbox{$f_j \in \cD$}, $j = 1, \dots , n$, 
in the states $\omega_\nu$, $\nu \in \NN$. 
Recalling that we may replace the scalar products
in equation \eqref{e.2.4} by their limit values,
we only need to consider the Fourier transforms of
\be
  u_1, \dots , u_n \mapsto
        e^{- \sum_j u_j \lambda_j  
        - (i/2) \sum_{k < l} u_k u_l \, {\mathrm Im}
    \langle f_k, f_l \rangle_\infty 
  - (1/2)  \sum_{k,l} u_k u_l \langle f_k , f_l \rangle_\infty} 
\ee
  at the points $l_\nu(f_1), \dots , l_\nu(f_n)$, $\nu \in \NN$. 
  Since the functions are elements of $L^1(\RR_+^n)$, their
  Fourier transforms are continuous and vanish  at
  infinity. Thus the limit state $\omega_\infty$ 
  can be defined on any sum of products of 
  resolvents as follows: there exists for the functions
  underlying the resolvents 
  a subset $\II \subset \NN$ such that the sequences
  $c(l_\iota(f_j))$, $\iota \in \II$, converge to
  $c(l_\infty(f_j))$, $j = 1, \dots , n$. 
  In view of the preceding results it is then clear that the
  sub-sequence of states $\omega_\iota$, $\iota \in \II$, 
  converges on the products of resolvents to the
  expression given in equation \eqref{e.2.7} whenever 
  $f_1, \dots , f_n \in \cD_\infty$. If
  $f_k \in L^2(\RR^s) \backslash \cD_\infty$
  for some $k$, the corresponding sequence converges to $0$. Thus
  $\omega_\infty$ is a limit point of the sequence
  $\omega_\nu$, $\nu \in \NN$, with
  properties stated in (ii).

  As to part (iii), the statement follows from (i) for gauge    
  invariant sequences of states. In the non-gauge invariant case
  we can proceed as in the proof of part~(ii). Let 
  $\cD_{K \infty} \doteq K \bigcap \cD_\infty$. 
  Since this space is finite dimensional, there exists a
  subset $\II \subset \NN$ such that
  $\lim_{\, \iota \in \II} l_\iota(g) = l_\infty(g)$ for 
  $g \in \cD_{K \infty}$; if $g \in K \backslash \cD_{K \infty}$,
  the sequence $l_\iota(g)$, $\iota \in \II$, diverges.  
  Now let $g_j \in K$, $j = 1, \dots , n$, be an arbitrary finite
  number of functions and consider an $n$-fold product of resolvents
  involving these functions. It follows from the preceding arguments
  that the expectation value of this product in
  the sequence of states $\omega_\nu$, $\nu \in \NN$, vanishes in the limit 
  if $g_k \in K \backslash \cD_{K \infty}$ for some $k$.
  On the other hand, the subsequence of states $\omega_\iota$, $\iota \in \II$,
  converges according to the preceding arguments
  to $\omega_\infty$ on all products of resolvents
  with $g_j \in \cD_{K \infty}$, $j = 1, \dots , n$. This completes the
  proof of the statement. 
\end{proof}

The limit states $\omega_\infty$, established in this proposition, are defined
on the full resolvent algebra $\fR$ and thus on its observable
and field subalgebras $\fA$, respectively~$\fF$. As we have seen, these 
states induce in general non-faithful GNS representations, where
certain specific resolvents are trivially represented. Yet, as can be
inferred from the relations given in the proposition,
the non-trivially represented 
subalgebras $\fR(\cD)$ in case (i) and $\fR(\cD_\infty)$
in case (ii) are regular, \ie the underlying resolvents with
functions $f$ in the respective subspaces of $L^2(\RR^s)$ have trivial kernels. 
This is a consequence of the fact that the expectation \mbox{values}  
of $i\lambda R(\lambda, f)$ converge to $1$ for $\lambda$ tending
to infinity, cf.\  \cite[Prop.~4.5]{BuGr2}.
It is this feature which allows for the continuous unitary
implementation of symmetry transformations of $\fR$ in the
limit representations in cases, where it fails for the Weyl
algebra.

\section{Trapped and untrapped equilibrium states}
\setcounter{equation}{0}

In this section we analyze the properties of Gibbs-von Neumann    
states on the resolvent algebra $\fR$ and of
their thermodynamic and infinite particle number limits. To this
end we proceed for given length $L > 0$, from the
single particle Hamiltonians $h_L \doteq \bP^2 + V_L(\bQ)$, defined on
a suitable domain in $L^2(\RR^s)$. Here $\bP$ denotes the momentum operator,
$\bQ$ the position operator and
\mbox{$\bx \mapsto V_L(\bx) \doteq L^{-2} \, V(\bx/L)$}
is a (scaled) potential, where $V$ is continuous and non-negative,
tending to infinity for large $|\bx|$. To simplify the
discussion, we assume that $L \mapsto V_L(\bx)$
is monotonically decreasing for fixed $\bx \in \RR^s$. 

\medskip
In order to have control
on the spectral properties of $h_L$, we assume that the  
potentials comply with the bounds required in \cite[Thm.\ XIII.81]{ReSi4}.
Simple examples are  $\bx \mapsto L^{-2 - \eta} \, |\bx|^\eta$
for $\eta > 1$.  Thus  $h_L$ is non-negative 
and has discrete spectrum. It is also apparent 
that each $h_L$ is mapped to $(1/L^2) \,  h_1$ by the 
adjoint action of a unitary scale transformation
depending on  $L$.
The ground state energy of $h_L$ is denoted by
$\epsilon_{L,1} = L^{-2} \epsilon_1$,
where $\epsilon_1$ is the positive, non-degenerate ground state energy
of $h_1$ \cite[Thm.\ XIII.47]{ReSi4}.
Applying the estimates on the asymptotic number of
eigenstates of $h_1$, given in the above reference,
all single particle partition functions
\be
\beta \mapsto \text{Tr}_{\cF_1} \, e^{- \beta \, h_L} = 
\text{Tr}_{\cF_1} \, e^{- (\beta / L^2) \, h_1} \, , \quad \beta >  0 \, , 
\ee
turn out to be finite.

\medskip
If $H_L$ is the second quantization of
$h_L$, often denoted by $d \Gamma(h_L)$, it follows from 
standard arguments that the operators
$e^{-\beta (H_L - \mu N)}$ are of trace class on $\cF$ 
for arbitrary inverse temperature $\beta > 0$ and
chemical potential $\mu < \epsilon_{L, 1}$.
Thus, putting $Z = \text{Tr}_\cF \, e^{-\beta (H_L - \mu N)}$, the
corresponding states on the resolvent algebra, 
\be \label{e.3.2}
\omega_{\beta, \mu, L}(R) \doteq
Z^{-1} \, \text{Tr}_\cF \, e^{- \beta (H_L - \mu N)} R \, , \quad R \in \fR \, ,
\ee
are trapped, gauge invariant equilibrium states     
(Gibbs-von Neumann states). 
They satisfy the KMS condition on $\fR$ with regard to the
adjoint action of the unitary group $t \mapsto e^{it(H_L - \mu N)}$
for the given $\beta$ and $\mu$. Note that these actions
coincide on the gauge invariant subalgebra $\fA \subset \fR$ for all 
values of the chemical potential~$\mu$.

\medskip
It is well known that these equilibrium states determine
quasifree states on the Weyl algebra and, since they are regular,
also on the resolvent algebra. As a matter of fact, any KMS state
on the Weyl algebra for given non-interacting time evolution is a
quasifree state under quite general conditions \cite{RoSiTe}.
In the case at hand the equilibrium states are fixed by the scalar
products, $ f,g  \in L^2(\RR^s) $,
\be \label{e.3.3} 
\langle f, g \rangle_{\beta, \mu, L} =
(1/2) \big( \langle f ,  e^{\beta (h_L - \mu)}
(e^{\beta (h_L - \mu)} -1)^{-1} g \rangle +
\langle g, \, (e^{\beta (h_L - \mu)} -1)^{-1} f \rangle \big) \, .
\ee

The thermodynamic limit is obtained by letting $L$ tend to infinity,
where the limit of the single particle Hamiltonians is denoted by
$h_\infty = \bP^2$ and the ground state energies are replaced by
their limit $\epsilon_{\infty , 1} = 0$. 
The subsequent lemma enters in our proof that   
the limits of certain basic sequences of equilibrium states exist
on the resolvent algebra and satisfy the KMS condition 
for the corresponding limit dynamics. 
\begin{lemma} \label{l.3.1} 
  Let $\langle \, \cdot \, , \, \cdot \, \rangle_{\beta, \mu, L}$ be the
  scalar products defined above.
\begin{itemize}
\item[(i)]
  Let $\beta > 0$, $\mu < 0$, and $t \in \RR$. Then
  \be \label{e.3.4}
  \lim_{L \rightarrow \infty} \, \langle f , e^{it(h_L - \mu)} g \rangle_{\beta, \mu, L}
  = \langle f , e^{it(h_\infty - \mu)} g \rangle_{\beta, \mu, \infty} \, , 
  \quad f,g \in L^2(\RR^s) \, .
  \ee
\item[(ii)]  Let $L \in \RR_+ \cup +\infty$. There exists a complex
  subspace $\cD_L \subset L^2(\RR^s)$, which is stable under
  the action of $e^{ith_L}$, $t \in \RR$, such that the limits
  \be \label{e.3.5} 
  \lim_{\mu \nearrow \epsilon_{L,1}}  \, \langle f ,
  e^{it(h_L - \mu)} g \rangle_{\beta, \mu, L}
  \doteq \langle f ,  e^{it(h_L - \epsilon_{L,1})} g
  \rangle_{\beta, \epsilon_{L,1}, L} \, , \quad f \in \cD_L \, , 
  \ee
exist and are continuous in $t$. On the other hand, 
  \be 
  \lim_{\mu \nearrow \epsilon_{L,1}} \,
  \langle f , f \rangle_{\beta, \mu, L}
  = + \infty \, , \quad f \in L^2(\RR^s) \backslash \cD_L  \, .
  \ee
  If $L < \infty$,  
  $\cD_L = (1 - E_{L,1}) \, L^2(\RR^s)$, where $E_{L,1}$ is the
  projection onto the ground state of $h_L$. If $L = + \infty$, 
  $\cD_L$ is the domain of~$|\bP|^{-1}$. 
  \end{itemize}
\end{lemma}
\begin{proof}
  (i) The monotonicity of the potentials      
  implies that $V_{L_1} \geq V_{L_2} \geq 0$, hence
  $(h_{L_1} - \mu) \geq (h_{L_2} - \mu) > (h_\infty - \mu) $ 
  if $L_2 \geq L_1$. Proceeding to the inverses, the
  sequence of resolvents $(h_L - \mu)^{-1}$ is
  monotonically increasing with increasing $L$ and
  converges to $(h_\infty - \mu)^{-1}$ in the strong
  operator topology. Now the functions
  $\epsilon \mapsto e^{\beta (\epsilon - \mu)} 
  (e^{\beta(\epsilon - \mu)} - 1)^{-1} e^{it(\epsilon - \mu)} $
  and $\epsilon \mapsto
  e^{it(\epsilon - \mu)} (e^{\beta(\epsilon - \mu)} - 1)^{-1}$
  are bounded and continuous on the closure of 
  $\RR_+$ since $\mu < 0$.
  It therefore follows from functional calculus 
  that the corresponding sequences of operators,
  where $\epsilon$ is replaced by~$h_L$, converges
  for large $L$ in the strong
  operator topology on $L^2(\RR^s)$
  to the operator, where $\epsilon$ is replaced
  by~$h_\infty$. This completes the proof of part~(i).
  
  \medskip
  (ii) Let $L$ be finite.    
  Since the ground state of $h_L$ is simple
  and isolated from the rest of the spectrum, there 
  is some $\delta > 0$ such that for $\mu$ in a neighborhood of 
  $\epsilon_{L,1}$ one has
  $(h_L - \mu) \, (1 - E_{L,1}) \geq \delta \, (1 - E_{L,1})$. Hence 
  the operator functions 
  $\mu \mapsto  (e^{\beta (h_L - \mu)} - 1)^{-1} (1 - E_{L,1})$
  as well as
  $\mu \mapsto e^{\beta (h_L - \mu)}  (e^{\beta (h_L - \mu)} - 1)^{-1} (1 - E_{L,1})$
  are norm continuous at $\mu = \epsilon_{L,1}$. So 
  the scalar products in
  equation \eqref{e.3.5} converge for $f,g \in (1 - E_1) L^2(\RR^s)$. 
  On the other hand,
  given any $f \in L^2(\RR^s)$ such that $E_{L,1} \, f \neq 0$ one has
  \be
  \langle f, (e^{\beta (h_L - \mu)} - 1)^{-1} f \rangle
  \geq (e^{\beta (\epsilon_{L,1} - \mu)} - 1)^{-1} \langle f, E_{L,1} f \rangle \, ,
  \ee
  which tends to $+\infty$ in the limit $\mu  \nearrow \epsilon_{L,1}$. Thus
  $\cD_L = (1 - E_{L,1}) \, L^2(\RR^s)$, 
  being stable under the action of $e^{ith_L}$, $t \in \RR$,
  is the maximal complex subspace on which relation \eqref{e.3.5}
  holds.

  Next, let $L = + \infty$. The operator function
  $\mu \mapsto 
  (e^{\beta (h_\infty - \mu)} - 1)^{-1}  h_\infty$
  converges in the limit $\mu \nearrow 0$
  in the strong operator topology to the bounded operator
  $(e^{\beta h_\infty} - 1)^{-1}  h_\infty $. So 
  for $f,g$ in the domain of $|\bP|^{-1} = h_\infty^{- 1/2}$ the
  scalar products in \eqref{e.3.5} also converge in this
  limit. On the other hand, if
  $f$ does not lie in this domain, it follows from
  the lower bound of the function
  \be
  \epsilon \mapsto
  (e^{\beta \epsilon} - 1)^{-1} \geq (\beta \epsilon )^{-1}
  e^{-\beta \epsilon } \, , \quad  \epsilon > 0 \, ,
  \ee
  that the expectation values
  $\langle f, (e^{\beta (h_\infty - \mu)} - 1)^{-1} f \rangle$ diverge 
  in the limit $\mu \nearrow 0$. Thus $\cD_L$ is equal to the
  domain of $| \bP |^{-1}$, completing the proof.  
\end{proof}
  
With the help of this lemma and results obtained in the preceding section
we arrive at the following statement on pertinent 
limits of equilibrium states. 

\begin{proposition} \label{p.3.2} 
  Let 
  $\omega_{\beta, \mu, L}$ be the equilibrium states on the resolvent 
  algebra $\fR$, defined in equation \eqref{e.3.2} for $\beta > 0$,
  $\mu < \epsilon_{L,1}$ and $L > 0$.
  \begin{itemize}
  \item[(i)] Let $\beta > 0, \, \mu < 0$. The thermodynamic limit of
    the states exists, pointwise on $\fR$,
    \be
    \lim_{L \rightarrow \infty} \omega_{\beta, \mu, L}(R) =
    \omega_{\beta, \mu, \infty}(R) \, , \quad R \in \fR \, . 
    \ee
    The limit states $\omega_{\beta, \mu, \infty}$ satisfy the 
    KMS condition at inverse temperature $\beta$ and
    chemical potential $\mu$ for the dynamics induced by the
    adjoint action of the unitary group $t \mapsto e^{it(H_\infty - \mu N)}$
    on $\fR$, where $H_\infty = d\Gamma(h_\infty)$.
  \item[(ii)] Let $L \in \RR_+ \cup \, +\infty$ be fixed and
    let $\beta > 0, \, \mu < \epsilon_{L,1}$. The infinite 
    particle number limit of the states exists, pointwise on $\fR$,
    \be
    \lim_{\mu \nearrow \epsilon_{L,1}} \omega_{\beta, \mu, L}(R) =
    \omega_{\beta, \epsilon_{L,1}, L}(R) \, , \quad R \in \fR \, .
    \ee
    The limit state $\omega_{\beta, \epsilon_{L,1}, L}$ satisfies the 
    KMS condition at inverse temperature $\beta$ 
    for the dynamics induced by the
    adjoint action of the unitary group
    \mbox{$t \mapsto e^{it(H_L - \epsilon_{L,1} N)}$} 
    on $\fR$. The subalgebra $\fR(\cD_L) \subset \fR$, which is 
    generated by resolvents assigned to functions in the complex subspace
    $\cD_L \subset L^2(\RR^s)$, defined in the preceding lemma,
    is stable under the dynamics. Its GNS representation, induced 
    by the limit state, is regular. All resolvents assigned to elements of
    $L^2(\RR^s) \backslash \cD_L$ are trivially represented there
    and thus generate an ideal of $\fR$.
  \item[(iii)] 
    The states $\omega_{\beta, \mu, \infty}$, $\mu \leq 0$, are
    mixing, \viz for any $R_1, R_2 \in \fR$ one has
    \be
    \lim_{t \rightarrow \infty}
    \omega_{\beta, \mu, \infty}(R_1 \, \ad{e^{it(H_\infty - \mu N)}}(R_2))
      = \omega_{\beta, \mu, \infty}(R_1) \, \omega_{\beta, \mu, \infty}(R_2) \, .
    \ee
    So these states describe pure phases. 
  \end{itemize}
\end{proposition}
\noindent{\bf Remark:} The limit states on $\fR$ in this    
proposition are not normal relative to the Fock representation.
But the adjoint action of the given unitary time translations on
$\cF$ leaves $\fR$ invariant,
thereby defining automorphisms of this algebra.
The action of these automorphisms on $\fR$ is implemented in the
limit representations by the adjoint action of modified unitary 
operators, generated by Liouvillians which depend on
the thermal parameters.

  \begin{proof}
    (i) The existence of the limit follows from the first parts of
    Lemma \ref{l.3.1} and Proposition \ref{p.2.1}; note that
    the scalar products in the lemma converge on all of $L^2(\RR^s)$. 
    The proof that the limit states satisfy the KMS condition
    is obtained by standard arguments. For subsequent reference,
    we sketch it here for arbitrary unitary groups 
    $t \mapsto e^{it h}$ on  $L^2(\RR^s)$
    with positive generator $h$, leaving some complex 
    subspace \mbox{$\cD \subset L^2(\RR^s)$} invariant. 
    Given $\beta > 0$, the corresponding scalar products are 
    \be
    \langle f, g \rangle_\beta =
    (1/2) \big( \langle f, e^{\beta h} (e^{\beta h} - 1)^{-1} g \rangle
    + \langle g, (e^{\beta h} - 1)^{-1} f \rangle \big) \, ,
    \quad f,g \in \cD \, .
    \ee
    Thus the functions $t \mapsto \langle f, e^{ith} g \rangle_\beta$
    can continuously be extended to the strip
    \mbox{$\{ z \in \CC : 0 \leq \Imm{z} \leq \beta \}$},
    are analytic in its interior, and bounded. Their boundary
    values at the upper rim of the strip are 
    $t \mapsto \langle f, e^{(it - \beta) h } g \rangle_\beta =
    \langle e^{ith} g, f \rangle_\beta$. So these
    scalar products, representing the
    two-point function in the corresponding quasifree state,
    satisfy the KMS condition. Making use of this information in
    equation~\eqref{e.2.4}, cf.\ also equation \eqref{e.kms},
    it follows once again
    from the dominated convergence theorem that
    the correlation functions of products of resolvents in
    $\fR(\cD)$ are continuous in time
    and satisfy the KMS condition at inverse temperature~$\beta$.
    Since the  sums of products of resolvents are stable
    under the action of the dynamics and dense in~$\fR$, 
    the KMS condition holds on this algebra. 
    In the case at hand $\cD = L^2(\RR^s)$,
    completing the proof of the first part. 

\medskip
(ii) The convergence of the functionals now follows from the second   
part of Lemma~\ref{l.3.1} and the first part of Proposition \ref{p.2.1}. 
Moreover, resolvents containing functions
$f \in L^2(\RR^s) \backslash \cD_L$  are trivially represented in
the GNS representations induced by the limit state. So they 
generate an ideal in $\fR$. The subalgebra 
$\fR(\cD_L) \subset \fR$ is regularly
represented, cf.\ the remarks at the end of Sect.~2.
Since $\cD_L$ is stable under the limit dynamics, products of  
resolvents containing a
factor with a function $f \in L^2(\RR^s) \backslash \cD_L$
are annihilated in the state at all times. The KMS condition
is then trivially satisfied. For products containing 
only resolvents with functions in $\cD_L$,  
the KMS property follows from the preceding argument.

\medskip
(iii) Given $\mu \leq 0$, let $\cD_{\mu, \infty} \subset L^2(\RR^s)$ be    
the subspace fixed by the state $\omega_{\beta, \mu, \infty}$, cf.\ the
preceding steps.  The statement
holds trivially if in expectation values of products of
resolvents in this state one of the underlying functions is not contained
in $\cD_{\mu, \infty}$. If $f,g \in \cD_{\mu, \infty}$ it follows from the
familiar spectral properties of $h_\infty = \bP^2$ that $\lim_{t \rightarrow \infty}
\langle f, e^{it(h_\infty - \mu)} g \rangle_{\beta, \mu} = 0$.
Making use of this information in relation \eqref{e.2.4}
one sees that products of resolvents at different
times in the state, involving only functions
in $\cD_{\mu, \infty}$, become uncorrelated in the limit of large time
differences. This completes the proof of the statement.
\end{proof}

  In the remainder of this section we    
  construct states out of the preceding ones by adding
  to the quantum field some classical background field.
  Let $f \mapsto l(f)$ be a 
  real linear functional on $L^2(\RR^s)$. Plugging this 
  functional into equation~\eqref{e.2.4}, where the scalar product 
  is given by equation~\eqref{e.3.3}, one obtains quasifree states
  $\omega_{\beta, \mu, L \, ; \, l}$ on the resolvent algebra.
  They are normal with regard to the Fock \mbox{representation} if
  $l(f) = \Imm{\langle e, f \rangle}$ for
  some $e \in L^2(\RR^s)$ and are  
  obtained by composing the KMS states
  $\omega_{\beta, \mu, L}$ with the adjoint action
  $\ad{W(e)} \doteq W(e)^* \, \cdot \, W(e)$ of the Weyl operator
  corresponding to $e$. 
  The resulting states are neither gauge invariant, nor stationary,
  and consequently do not satisfy the KMS condition.

\medskip 
Being interested in the phenomenon of condensation,   
we restrict our attention to certain specific functionals,
which are determined by eigenvectors of $h_L$. 
Let  $e_{L, k}$ be such an eigenvector corresponding to the eigenvalue
  $\epsilon_{L, k} = L^{-2} \epsilon_k$, where
  $\epsilon_k$ is the $k$-th eigenvalue of $h_1$, 
  and let $l_{L, k}(f) \doteq \Imm{\langle e_{L,k} , f \rangle}$,
  $k \in \NN$. The time dependence of the resulting expectation values
  in the states $\omega_{\beta, \mu, L \, ; \, l_{L,k}}$ 
  is conveniently determined for Weyl operators $W(f)$,
  cf.\ equation~\eqref{e.2.2}, 
  \be  \label{e.3.13} 
  t \mapsto  \omega_{\beta, \mu, L \, ; \, l_{L,k}} \big(
    \ad{e^{it(H_L - \mu N)}}\big( W(f) \big) \big) = 
    e^{i l_{L,k}(t)(f)} \, e^{-(1/2) \langle f, f \rangle_{\beta, \mu, L}} \, .
  \ee 
  Here
  \be
  l_{L,k}(t)(f) =
  \Imm{\langle e_{L,k} , e^{it(h_L - \mu)} f \rangle}
  = \Imm{e^{it(\epsilon_{L,k} - \mu)} \langle e_{L,k} , f \rangle} \, .
  \ee
  So the states change periodically in time with period
  $\tau \doteq 2 \pi / (\epsilon_{L,k} - \mu)$. 
  Taking a mean of these states over the time 
  interval $\tau$, one  obtains the averaged
  states 
  \be \label{e.3.15}
  \overline{\omega}_{\beta, \mu, L \, ; \, l_{L,k}} \doteq
  (1/\tau) \int_0^\tau \! dt \,
  \omega_{\beta, \mu, L \, ; \, l_{L,k}} \scirc \, \ad{e^{it(H_L - \mu N)}} \, .
  \ee
  In expectation values of Weyl operators,    
  the phase factor in equation \eqref{e.3.13} is then replaced
  by
  \be
  (1/\tau) \int_0^\tau \! dt \, e^{i l_{L,k}(t)(f)} \, = \, 
    J_0(\mbox{\footnotesize $| \langle e_{L,k} , f \rangle |$}) \, ,
  \ee
  $J_0$~being the Bessel function of first kind of
  order zero. In case of products of resolvents which are 
  evaluated in the
  averaged states $\overline{\omega}_{\beta, \mu, L \, ; \, l_{L,k}}$, 
  one has to replace the phase factor
  $e^{- i \, l_{L,k}( \sum_j u_j f_j)}$ appearing in relation~\eqref{e.2.4} by
  its time average $J_0(|\langle e_{L,k}, \sum_j u_j f_j \rangle |)$,
  called Bessel factor henceforth. 
  The resulting state is gauge invariant and
  stationary, but it does not satisfy the
  KMS condition, cf.\ Proposition \ref{p.3.3} below.
  One can remedy this deficiency by proceeding
  to the thermodynamic limit.

\medskip 
The time averages \eqref{e.3.15} of the states    
affect in general also their domains of regularity in the
thermodynamic limit. To see this, let $l_{L,k}$
be a (generalized) sequence of functionals 
which approximates $l_{\infty,k}$ in the limit of large $L$
and diverges for some $f \in L^2(\RR^s)$.
Then the resulting averaged state vanishes on products
of resolvents containing a resolvent
with a function $ \zeta f$ for some    
$\zeta \in \TT$. This follows from the decay properties
of the Bessel function for large arguments.
Thus the limit functional $l_{\infty, k}$ 
contributes to the thermodynamic limit of the averaged state
only for elements in the maximal complex subspace of its domain, 
denoted by~$\overline{\cD}_\infty \doteq \bigcap_{\, \zeta  \in \TT} \,
  \zeta \, \cD_\infty$. 
  
  \medskip
  We turn now to the determination of the thermodynamic as well as    
  infinite particle number limits of the states
  $\overline{\omega}_{\beta, \mu, L \, ; \, l_{L,k}}$. The underlying scalar
  products $\langle \, \cdot \, , \, \cdot \, \rangle_{\beta, \mu, L}$
  were already analyzed in Lemma \ref{l.3.1}. So
  we must only have a closer look at the functionals
  $f \mapsto l_{L,k}(f)$. As discussed in Sec.\ 2, the functionals 
  $l_{L,k}(f)$ approach their limits $l_{\infty,k}(f)$ in general only  
  in the weak topology, \viz for any given finite number of
  functions $f$ one has to proceed to specific subsequence 
  of the scaling factors $L_\iota$, \mbox{$\iota \in \II$}. 
  In order to retain control 
  of the continuity properties of the limit states with regard to the  
  time translations, they need to be adjusted
  in the  approximating correlation functions 
  according to the respective scale $L$.
  This is also natural from the point of view of physics   
  where one studies the dynamical properties of
  systems for differing external
  constraints, such as the length scale~$L$. 
  Applying this rule to the functionals, it follows from
  the equality 
  \be 
   l_{L,k}(e^{it h_{L,k}}f) =
  \cos(t \epsilon_k/L^2) \, l_{L,k}(f)
  + \sin(t \epsilon_k/L^2) \, l_{L,k}(i f) \, ,
  \quad f \in \overline{\cD}_\infty \, ,
    \ee
    that
    $\lim_{\iota \in \II}  l_{L_\iota,k}(e^{it h_{L_\iota,k}}f) =
    l_{\infty,k}(f)$, $t \in \RR$. So the limit of these particular 
    functionals, contributing to the correlation functions
    in the thermodynamic limit, turns out to be invariant
    under the action of the limit dynamics.
   
  \begin{proposition} \label{p.3.3}
    Let $\overline{\omega}_{\beta, \mu, L \, ; \, l_{L,k}}$ be the
    averaged states on the
    resolvent algebra $\fR$, which are determined by relations \eqref{e.3.13}
    and \eqref{e.3.15} for $\beta > 0$, $\mu < \epsilon_{L , 1}$, $L > 0$, and
    functionals $l_{L,k}$ specified above. 
    \begin{itemize}
    \item[(i)] Let $\beta > 0$, $\mu < 0$. There exists a state
      $\overline{\omega}_{\beta, \mu, \infty \, ; \, l_{\infty,k}}$ on $\fR$ 
      which is a \mbox{weak-*} limit point of the given states
      for large $L$. It has the following properties:
      there is a complex subspace
      $\overline{\cD}_\infty \subset L^2(\RR^s)$ such that 
      the algebra $\fR(\overline{\cD}_\infty)$ is regularly represented
      in the GNS representation induced by the limit state; 
      resolvents assigned to functions in
      $L^2(\RR^s) \backslash \overline{\cD}_\infty$ are
      trivially represented there.
      The state is gauge invariant and stationary with regard to 
      the limit dynamics $t \mapsto \ad{\,e^{it(H_\infty - \mu N)}}$,
      and the correlation functions depend continuously
      on time. If the underlying limit functional
      $l_{\infty, k}$ is different from $0$, the limit state does not
      satisfy the KMS condition.
      \item[(ii)] Let $L \in \RR_+ \, \cup \, +\infty$ be fixed and
     let $\beta > 0$, $\mu < \epsilon_{L,1}$. The infinite 
     particle number limit of the states exists, pointwise on $\fR$,
     \be
     \lim_{\mu \nearrow \epsilon_{L,1}}
     \overline{\omega}_{\beta, \mu, L \, ; \; l_{L, k}}(R) =
     \overline{\omega}_{\beta, \epsilon_{L,1}, L \, ; \, l_{L,k}}(R) \, ,
     \quad R \in \fR \, .
     \ee
     The algebra $\fR(\cD_L \bigcap \overline{\cD}_\infty)$ is regularly
     represented in the GNS representation induced by the limit
     state, where $\cD_L$ was defined in Lemma \ref{l.3.1}.
     Resolvents assigned to functions in
     $L^2(\RR^s) \backslash (\cD_L \bigcap \overline{\cD}_\infty)$ are
     trivially represented there.

     \medskip 
     At level $k = 1$, all limit states     
     satisfy the KMS condition at inverse temperature
     $\beta$ for the dynamics given by
     $t \mapsto \ad{e^{it(H_L - \epsilon_{L,1}N)}}$;
     if $L < \infty$ they are unique and do not depend on the
     normalization of the wave functions  $e_{L,1}$ entering in $l_{L,1}$.
     For level $k > 1$ and finite $L$
     the limit states do not satisfy the KMS condition.
     In the thermodynamic limit  $L = +\infty$, there exist 
     several limit states satisfying the KMS condition
     for any given level $k$; they depend on the 
     normalization of the wave functions  $e_{L,k}$ in the
     approximating functionals $l_{L,k}$, cf.\ the remark below. 
   \item[(iii)]
     Let $\beta > 0$, $\mu = 0$, and $l_{\infty,k} \neq 0$
     for some $k \in \NN$. The corresponding
     KMS~state \ $\overline{\omega}_{\beta, 0, \infty \, ; \, l_{\infty, k}}$ is not
     mixing, \ie describes a mixture of phas\-es. It can (centrally)
     be decomposed into a mean over the gauge group, 
    \be \label{e.central}
    \overline{\omega}_{\beta, 0, \infty \, ; \, l_{\infty, k}}
    =  (1/2 \pi) \int_0^{2 \pi} \! du \ 
   \omega_{\beta, 0, \infty \, ; \, l_{\infty, k}} \, \scirc \, 
   \gamma_u \, ,
   \ee
   where $\omega_{\beta, 0, \infty \, ; \, l_{\infty, k}}$ is a KMS state,
   describing a pure phase. This state is not gauge invariant.
     \end{itemize}
  \end{proposition}

  \noindent \textbf{Remark:} If $L = \infty$,
  the space $\overline{\cD}_\infty$,    
  determined by the limit functional $l_{\infty, k}$, 
  must have a non-trivial intersection with 
  the domain of~$|\bP|^{-1}$ in order to arrive at non-trivial states
  in the limit of infinite particle number.
  Let us briefly comment on this point 
  in case of potentials $\bx \mapsto V(\bx)$ which are
  real analytic. There the corresponding (real) functions
  $\bx \mapsto e_k(\bx)$ of $h_1$ are also real analytic, bounded, 
  and the complement of their nodal sets consists of at most
  $k$ connected components, cf.\ \cite[Ch.\ 3.4]{BeSh}.  
  If $k=1$, the function $\bx \mapsto e_1(\bx)$
  can be chosen to be positive, so for
  test functions $f$ one obtains in the limit  
  $l_{\infty ,1}(f) = e_1(0)  \int \! d\bx \, \mbox{Im} f(\bx) $.
  In $s > 2$ dimensions, any test function lies in the
  domain of $|\bP|^{-1}$ and the limit functional is
  non-trivial. But in $s = 1,2$ dimensions, this domain
  condition implies that $f$ must be the (partial)
  derivative of another test function, 
  $f = - \bpartial \bg$, so the integral vanishes.
  In that case one can renormalize the approximating
  functionals $l_{L,1}$ by an additional factor $L$,
  giving in the limit
  $l'_{\infty ,1}(f) = \bpartial e_1(0) \cdot 
  \int \! d\bx \, \mbox{Im} \, \bg(\bx) $. 
  If this limit is still zero, one can continue this
  procedure with higher powers of $L$ and smaller subsets of
  test functions. It yields a non-trivial result after a finite
  number of steps since $e_1$ is real analytic. Thus also in
  $s=1,2$ dimensions the ground state wave functions gives rise
  to functionals in the thermodynamic
  limit with non-trivial domains $\overline{\cD}_\infty$. 
  In a similar manner the existence of functionals involving 
  the excited states $e_k$, $k > 1$, can be established, including
  also those cases, where the nodal sets pass through the origin.

\begin{proof}
  (i) According to the arguments given
  in the proof of the second part of Proposition
  \ref{p.2.1} and Lemma \ref{l.3.1},
  the thermodynamic limits of the time dependent correlation
  functions in the given states coincide 
  with those in the states
  $\overline{\omega}_{\beta, \mu, \infty \, ; \, l_{L,k}}$,
  where the underlying scalar products  
  $\langle \, \cdot \, , \, \cdot \, \rangle_{\beta, \mu, L}$
  are replaced by their limit
  $\langle \, \cdot \, , \, \cdot \, \rangle_{\beta, \mu, \infty}$.
  The $L$-dependence of the correlation functions in 
  the latter states arises from the term 
  \be \label{e.3.20}
  J_0(|\langle e_{L, k} , \sum_j u_j e^{it_j(h_L - \mu)} f_j \rangle | )
  = J_0(|\sum_j u_j e^{it_j(\epsilon_k/L^2 - \mu)} 
  \langle  e_{L, k},  f_j \rangle | ) \, ,
  \ee
where $f_j \in L^2(\RR^s)$, $t_j \in \RR$ for $j = 1, \dots , n$.
It replaces in  equation \eqref{e.2.4} the 
phase factor depending on the functionals $l_{L,k}$.
Since $\langle e_{L, k} , f \rangle
= (l_{L,k}(if) + i \, l_{L,k}(f) )$ it follows, as explained,
that the term \eqref{e.3.20}  vanishes in the
thermodynamic limit for almost all $u_1, \dots u_n$ if
$f_k \in L^2(\RR^s) \backslash \overline{\cD}_\infty$ for some $k$. 
If all functions are contained in $\overline{\cD}_\infty$,
the limit is given by $ J_0(|\sum_j u_j e^{-it_j \mu } 
\langle  e_{\infty, k},  f_j \rangle | ) $,
where $f \mapsto \langle  e_{\infty, k},  f_j \rangle$
is the complex linear functional on $\overline{\cD}_\infty$, 
determined  by $l_{\infty, k}$. The representation of the 
algebra $\fR(\cD_\infty)$ induced by the limit state
is regular, as can be seen by inspection of the expectation 
values of the underlying resolvents 
in equation \eqref{e.2.4}, cf.\ the remark at the end
of Sect.\ 2. It is also clear that all resolvents
assigned to functions $f \in L^2(\RR^s) \backslash \overline{\cD}_\infty$
are trivially represented there.

\medskip
The limit state is gauge invariant, being approximated    
by such states. Putting $t_j = t$, $j = 1, \dots , n$, in
the Bessel factor and the scalar products
appearing in the correlation functions of the
limit state, derived from \eqref{e.2.4}, it is apparent that this state is
stationary under the action of the limit dynamics. 
Moreover, the continuity properties of the
time translations on $L^2(\RR^s)$ and of  
the Bessel factor imply that  
the correlation functions are continuous with
regard to separate time translations of the underlying operators.
But if $l_{\infty, k} \neq 0$ and $\mu < 0$, 
the time dependence of the 
Bessel factor impedes the KMS property of the limit state. 

\medskip 
(ii) Adopting the
arguments given in the proof of Proposition \ref{p.3.2}, 
the convergence of the states follows from the
convergence, respectively divergence properties 
of the scalar products, established in
Lemma \ref{l.3.1}, and the continuity properties of the  
Bessel factor with regard to $\mu$.
The properties of the resulting algebras are then  
established as in the preceding step. 

\medskip
If $k = 1$, the Bessel factor is time independent for
$\mu = \epsilon_{L,1}$ and any finite or
infinite length $L$.
The KMS property of the resulting states therefore follows
from the corresponding properties of the
scalar products, established in the proof of
Proposition~\ref{p.3.2}. Now for finite $L$ the domain
$\cD_L$ is equal to $(1 - E_{L,1}) \, L^2(\RR^s)$, which
lies in the kernel of $l_{L,1}$, hence the KMS state
is not modified by this functional. If $k > 1$, the phase factors in
the Bessel factor \eqref{e.3.20} oscillate with frequency
$(\epsilon_k - \epsilon_1)/L^2$, destroying  
the KMS property. In the thermodynamic limit $L = +\infty$ and $\mu = 0$,
this factor is independent of time for any $k \in \NN$. So 
the KMS property of the states follows again from the KMS property
of the scalar products. Examples of non-trivial
limit functionals which modify the states were presented
in the preceding remark. 

\medskip 
(iii) Proceeding as in the proof of the third part of Proposition   
\ref{p.3.2}, there appears in relation \eqref{e.2.4}
the Bessel factor
$ J_0(|\sum_j u_j \langle  e_{\infty, k},  f_j \rangle |) $ 
in the expectation values of products of resolvents at
different times in the averaged state.  
It is also time independent. But, in contrast to the
exponential function, it does not split into
a product of factors containing only functions
entering in corresponding factors of the product of resolvents.
This implies after a moments reflection that the state is not mixing. 
The decomposition of the state into pure phases is 
accomplished by decomposing the Bessel factor into 
its exponential contributions, involving the gauge transformed
test functions, 
\be 
J_0(|\sum_j u_j \langle  e_{\infty, k},  f_j \rangle |)
= (1/2 \pi) \! \int_0^{2 \pi} \! \! du \,
e^{i \sum_j u_j \, l_{\infty,k}(e^{iu} f_j)},
\ \ f_j \in \cD_L \, 
{\textstyle \bigcap} \, \overline{\cD}_\infty \, .
\ee
This mean can be interchanged with the 
integration in equation \eqref{e.2.4}
for the averaged state. Since the scalar
products are invariant under simultaneous gauge transformations of its
entries one arrives at equality \eqref{e.central}.
By arguments given in the proof of the third
part of Proposition \ref{p.3.2}
one sees that the state $\omega_{\beta, 0, \infty \, ; \, l_{\infty, k}}$,
appearing on the right hand side of this equality, is mixing. Since
$l_{\infty, k} \neq 0$ it is not gauge invariant, however.
\end{proof}
This proposition shows 
that equilibrium states of physical interest,    
which are limits of Gibbs-von Neumann states, 
can be defined on the resolvent algebra without running into any
mathematical problems. The results shed light on some 
basic points. First, the infinite particle number limit of
trapped equilibrium states exists for any value of the temperature
and dimension. The
resulting equilibrium states include an infinity of
particles in the ground state which, as we shall see in more detail,
describe a  condensate. The preceding results
imply that this condensate cannot be changed by adding 
to it more of these particles with the help
of Weyl operators, the limit states are
unique. Second, the thermodynamic limit of
trapped equilibrium states exists for any value of the
temperature, chemical potential and dimension of space. 
We will see that for vanishing chemical potential 
these states describe \textit{quasi condensates} 
of zero energy modes (improper states) in one
and two dimensions. In higher
dimensions the states exhibit a maximal (critical) local
density, however. Yet the preceding results imply that 
Weyl operators creating such modes act non-trivially on
the states. So one can increase their density unlimitedly,
leading to the formation of condensates.

\medskip
In common (textbook) discussions of the phenomenon of   
Bose-Einstein condensation 
one proceeds from the (grand) canonical
ensemble with given sharp (mean) particle number and sharp
boundaries (boxes). Anticipating that the states of interest are
homogeneous, the particle density is defined as the quotient
of the number of particles and the volume of the boxes.
In order to exhibit condensates, one then compares the density
of particles in the ground
state with the total density of particles. This
well founded approach obviously does not work in case of systems
held together by regular trapping potentials, which do not have well
defined boundaries. It has therefore been proposed to use
in these models other quantities. For example, instead of the
volume, the energy level density of the particle states  
in the potential is taken as a substitute. Such ideas have been
applied in numerous proficient publications, cf.\ \cite{LiSeSoYn}
and references quoted there. Nevertheless, they are conceptually not
fully satisfactory, cf.\ \cite[Sec.\ 4.8]{Ve}. 
For example, the quantities proposed as substitutes for the
particle density can not be regarded as local ordering parameters. In
particular, they defy an interpretation in terms of
expectation values of an observable which determines 
the particle densities at different points, expected 
to vary within trapped systems. This leads us to a  
study of the concept of local particle density in the framework of the
resolvent algebra. It is in our opinion a meaningful 
tool for the analysis of condensates in trapped and untrapped systems

\section{Particle densities}
\setcounter{equation}{0}

Heuristically, the observable
determining the particle density is represented by the
function $\bx \mapsto a^*(\bx) a(\bx)$, which can be given a 
rigorous meaning in the sense of operator valued
distributions in the Fock representation. Its integral
is the particle number operator~$N$. Since we are
dealing here also with disjoint representations, there arises
the question of how one can decide whether this density is still
meaningful there. It is another virtue of
the resolvent algebra that this physically important question
can be answered within its framework.

\medskip
As we have seen, the (limit) states $\omega$ of interest typically   
determine some complex subspace $\cD \subset L^2(\RR^s)$ for which the
corresponding resolvent algebra $\fR(\cD)$ is regularly
represented in the induced GNS representation; resolvents attached to
the complement of $\cD$ are trivially represented there.  
Now let $K \subset \cD$ be any finite dimensional complex subspace.
Then, \textit{a fortiori}, the restriction of the representation
to the subalgebra $\fR(K)$ is regular. This implies by the
Stone-von Neumann theorem that this restriction is quasi equivalent to the
Fock representation of~$\fR(K)$, cf.\ \cite[Thm.\ 4.5]{BuGr2}
and its subsequent remark.
In the Fock representation of $\fR(K)$, the particle number operator
$N_K$ is densely defined. Choosing some orthonormal basis
$e_1, \dots , e_n \in K$, it can be presented in the familiar form 
$N_K = \sum_{j=1}^n a^*(e_j) a(e_j)$. In fact,
it can be approximated on its domain  by
observables in $\fA(K) \subset \fR(K)$ in the strong operator
topology.  We briefly sketch the argument: given $f \in K$ and
$\varepsilon > 0$, the operators 
\be \label{e.4.1}
a^*(f) a(f) \, ( 1 + \varepsilon a^*(f) a(f) )^{-1} =
(1/\varepsilon) \, \big( 1 -  (1 + \varepsilon a^*(f) a(f) )^{-1} )
\ee
are elements of $\fA(K)$. This follows from \cite[Lem.\ 3.1]{Bu3}   
according to which the operators
$\big(1 + \varepsilon a^*(f) a(f)\big)^{-1}$ are contained in 
$\fR(K)$; in fact, they are 
members of some compact ideal of $\fR(K)$ and 
since they are gauge invariant, they are
contained in $\fA(K)$. It is also clear that
$a^*(f) a(f) \leq \| f \|^2 \, N_K$. Hence the 
operators \eqref{e.4.1}
converge in the strong operator topology on the domain of~$N_K$ to
the unbounded operator $a^*(f) a(f)$ in the limit $\varepsilon \searrow 0$.  

\medskip
We make use of this fact in order to determine the local particle   
properties of states $\omega$ on the corresponding regular subalgebras
$\fR(\cD) \subset \fR$. 
Let $\bO \subset \RR^s$ be any open bounded region
and let $\cD(\bO) \subset L^2(\bO) \bigcap \cD$ be the
subspace of functions having support in~$\bO$; depending on the
given state, it may happen that $\cD(\bO) = L^2(\bO)$. Since the
operator sequence \eqref{e.4.1} is monontonically increasing
for decreasing $\varepsilon$, the following definition is
meaningful. 

\bigskip \noindent
\textbf{Definition:}  Let $\omega$ be a state on   
the resolvent algebra which is regular on 
$\fR(\cD)$ and annihilates the ideal in $\fR$
which is generated by resolvents with 
functions in $L^2(\RR^s) \backslash \cD$. 
Let $\bO \subset \RR^s$ be an open bounded region
and let $\fR(\cD(\bO)) \subset \fR(\cD)$ be the
corresponding regular subalgebra. Then 
\begin{itemize}
\item[(i)]
  $\omega$ admits a (partial)
  particle interpretation in~$\bO$ if for any $f \in \cD(\bO)$
  \be \label{e.4.2}
  \omega(a^*(f) a(f)) \doteq
  \lim_{\varepsilon \searrow 0} \,
  \omega(a^*(f) a(f) \, ( 1 + \varepsilon \, a^*(f) a(f) )^{-1})
  < \infty  \, .
  \ee
  The number of particles in $\bO$ with a wave function 
  $f \in L^2(\bO) \backslash \cD(\bO)$ is not defined,
  \ie the corresponding approximations \eqref{e.4.2}
  tend to infinity. 
\item[(ii)] $\omega$ is locally normal on $\fR(\cD(\bO))$ if 
  it has a particle interpretation in $\bO$ 
  and, for some  orthonormal basis
  $e_j \in \cD(\bO)$, $j \in \NN$, one has
 \be
 \omega(N_\omega(\bO)) \doteq \lim_{n \rightarrow \infty} \,
 \sum_{j = 1}^n \omega(a^*(e_j) a(e_j)) < \infty \, .
 \ee
\end{itemize}

\vspace*{-3mm}
\noindent \textbf{Remark:}    
Here $N_\omega(\bO)$ is the \textit{regular} number operator counting  
particles in the region~$\bO$ with wave functions
in $\cD(\bO)$. This operator is defined in the GNS representation of
$\fR(\cD(\bO))$ induced by $\omega$. In fact, this representation is  
quasi equivalent to the Fock representation of this
algebra~\cite{DeDo}.

\medskip
The regular number operators $N_\omega(\bO)$   
in this definition are physically significant order 
parameters, allowing to establish the appearance, respectively
absence, of condensates in equilibrium states $\omega$
within the region $\bO$. In trapped systems,
described by Gibbs-von Neumann states $\omega$, they are defined
for functions in the largest possible space $\cD(\bO) = L^2(\bO)$,
hence $N_\omega(\bO) = N(\bO)$. But also in the
thermodynamic limit the resulting states are frequently
locally normal in this strong sense. Keeping the
temperature fixed and proceeding to the maximal possible value of the
chemical potential $\mu$, one can then check whether the
corresponding expectation values
of $N_{\omega_\mu}(\bO)$ stay finite, or tend to infinity. 
In the former case the limit system has a critical density
in~$\bO$ at the given temperature. In the latter case, there
are no limitations on the density, signaling condensation 
as we shall see. 

\medskip
In computations of the expectation values of   
local particle number operators it is
sometimes  convenient to proceed to their
local densities, instead of controlling infinite
sums of occupation numbers. We therefore
recall here briefly this standard device. 
Let $d_\varepsilon$, $\varepsilon > 0$, be a sequence of test functions 
on $\RR^s$ which converges to the
Dirac measure at the origin if $\varepsilon$ tends to $0$  
and let $h$ be a positive test function with compact support. 
As is well known, one can proceed  in Fock space to the limit 
\be
\lim_{\varepsilon \searrow 0} \int \! d\bx \, h(\bx) \,
\ad{U(\bx)} (a^*(d_\varepsilon) a(d_\varepsilon))
\doteq  \int \! d\bx \, h(\bx) \, a^*(\bx) a(\bx) \, ,
\ee
where one has convergence
in the sense of sesquilinear forms between
vectors with finite particle number and regular
wave functions. The resulting form
can be extended to a selfadjoint operator and if
$\bO_1 \subset \text{supp} \, h \subset \bO_2$ one obtains
for it the bounds
\be
\inf_{\sbx \in \sbO_1} \! \! h(\bx) \, N(\bO_1)
\leq \int \! d\bx \, h(\bx) \, a^*(\bx) a(\bx) \leq
\sup_{\sbx \in \sbO_2} \! \! h(\bx) \, N(\bO_2) \, .
\ee
Similar formulas can be established for the regular  
number operators in non-Fock states. Whereas the assignment
$\bx \mapsto \omega(a^*(\bx) a(\bx))$ is in general only defined 
in the sense of distributions, the resulting expectation
values can often be presented by continuous functions.
The desired information about condensation can then
directly be extracted from the respective two-point functions,
as we will see in the subsequent section.

\section{Probing condensation}
\setcounter{equation}{0}

We analyze now the trapped equilibrium states which we have   
constructed, making use of the local order parameters
put forward in the preceding section. In order to 
highlight the virtues of our approach, we begin by 
considering the simple example of harmonic trapping potentials. Let
$h_1 = (\bP^2 + \bQ^2)$ be the unscaled Hamiltonian
on $\RR^s$. Its normalized ground state is denoted by
\mbox{$e_1 \in L^2(\RR^s)$}, carrying  the energy
$\epsilon_1 = s$. The corresponding equilibrium states
$\omega_{\beta,\mu,1}$ 
on the resolvent algebra for given 
$\beta > 0$ and $\mu < \epsilon_1 = s$
are fixed by equation \eqref{e.3.3} with $L=1$. 
We want to determine in the limit of maximal
chemical potential $\mu$ the expected number of particles
with given wave function in $\cD = (1 - E_1) L^2(\RR^s)$, where $E_1$
is the projection onto the ground state of $h_1$.
For $f,g \in L^2(\RR^s)$ we obtain in the limit
$\mu \nearrow \epsilon_1$, recalling that the limit state
is regular on $\fR(\cD)$, 
\be \label{e.extension}
\omega_{\beta, \epsilon_1, 1}(a^*((1-E_1)f)a((1-E_1)g))
= \langle g, (e^{\beta(h_1 - \epsilon_1)} - 1)^{-1} (1 - E_1) f \rangle \, .
\ee
It follows from this equality that we can extend the limit
state on $\fR(\cD)$ to a quasifree state
on the full resolvent algebra $\fR$, denoted by
$\bomega_{\beta, \epsilon_1, 1}$, which
does not contain a single particle with wave function $e_1$. This 
fact allows us to define its particle density, as outlined in the
preceding section: we have 
\be \label{e.5.1}  
  \langle g, (e^{\beta(h_1 - \epsilon_1)} - 1)^{-1} (1 - E_1) f \rangle 
  = 
  \sum_{n = 1}^\infty \ 
  \langle g, e^{-n \beta(h_1 - \epsilon_1)} (1 - E_1) f \rangle \, . \nonumber
\ee
The kernels of the operators $ e^{-n \beta(h_1 - \epsilon_1)} (1 - E_1)$
on $L^2(\RR^s)$ are given by 
\begin{align} 
  & \bx, \by \mapsto 
  \langle \bx |  e^{-n \beta(h_1 - \epsilon_1)}
  (1 - E_1) | \by \rangle \nonumber \\
&  = e^{n \beta \epsilon_1} \, (2 \pi \sinh(2n\beta))^{-s/2} \,
  e^{-(1/2) \, (\coth(2\beta n)(\sbx^2 + \sby^2) - 2 \sbx \sby /\sinh(2 \beta n))} \nonumber \\
& -  \pi^{-s/2} e^{-(\sbx^2 +\sby^2)/2} \, , 
\end{align}
where in the second line the Mehler formula has been used; the 
contribution in the last line
is due to the projection. The sum
of these kernels is absolutely convergent 
and the result is continuous in $\bx, \by$. In particular, the
regular particle density in the (extended) limit state is given by
\begin{align} \label{e.critical}
  & \bx \mapsto
  \bomega_{\beta, \epsilon_1, 1}(a^*(\bx) a(\bx))  \nonumber \\
  & = \pi^{-s/2}  \sum_{n=1}^\infty \, \big((1 - e^{-4n\beta})^{-s/2} \,
  e^{(1 - \coth(2n\beta) + 1/\sinh(2n\beta)) \, \sbx^2}  - 1 \big) \, e^{-\sbx^2} \, .
\end{align}
It implies
$\omega_{\beta, \epsilon_1, 1}(N_{\omega}(\bO))
\leq \int_\sbO d\bx \, 
\bomega_{\beta, \epsilon_1, 1}(a^*(\bx) a(\bx))$, where
the inequality is due to the completion of the basis
in $\cD(\bO)$ to a basis in $L^2(\bO)$. One obtains 
equality in the limit $\bO \nearrow \RR^s$.
So for all temperatures, the regular local observables in 
the trapped equilibrium states indicate a
maximal (critical) local density of the thermal cloud.
On the other hand, the number of particles in the
approximating states $\omega_{\beta,\mu,1}$ with a wave function 
$f$ having some overlap with the ground state wave function,
$E_1 f \neq 0$, diverges in the limit. So we
conclude that the particles in the ground state form
condensates in the approximating states whenever the
total expected number of particles in the trap exceeds the
critical number $N_\omega$ of particles in the thermal cloud. By
comparing the density of all particles with the   
critical density \eqref{e.critical} of the cloud, one can 
determine the amount of condensate which is formed locally.

\medskip 
As was already mentioned, one can increase particle densities   
by composing states with Weyl automorphisms. In order to
arrive at stationary states, one must choose
eigenfunctions of $h_1$ in the underlying Weyl operators.
According to  Proposition~\ref{p.3.3}, choosing 
the ground state wave function $e_1$ and taking
a time average does not alter the limit states.
For (arbitrarily normalized) excited states $e_k$ 
one obtains, bearing in mind that the original state is 
gauge invariant, 
\be \label{e.5.4}
\omega_{\beta,\mu,1} \, \scirc \, \ad{W(e_k)} \, \big(a^*(f) a(f) \big)
= \omega_{\beta,\mu,1} (a^*(f) a(f)) +  |\langle e_k, f \rangle|^2 \, .
\ee
In this manner, the density of particles with
wave function $e_k$ can be made arbitrarily large in
any given region $\bO$. Yet, as we have seen
in Proposition \ref{p.3.3}(ii), the resulting states
as well as their time averages are not
in equilibrium since the correlations between operators at different times
exhibit oscillations which violate the KMS condition.
Due to the lack of interaction, these
oscillations are not suppressed by time averages,
the added particles do not equilibrate in the \mbox{thermal} background
and feel the external forces forever.  
It is an interesting question whether this feature disappears in
trapped interacting theories.

\medskip
For systems of particles trapped by non-harmonic forces, there    
are no such simple formulas for the density of particles
with regular wave functions. Yet one can rely there on the
method of counting the number of particles
locally. This amounts to computing
for bounded regions \mbox{$\bO \subset \RR^s$} and corresponding
projections $E(\bO)$ onto $L^2(\bO) \subset L^2(\RR^s)$
the expectation values of the corresponding regular number
operators $N_\omega(\bO)$, 
\be
\omega_{\beta, \epsilon_1, 1}(N_\omega(\bO))
\leq 
\text{Tr} \,
E(\bO) \, (e^{\beta (h_1 - \epsilon_1)} - 1)^{-1} (1 - E_1) \, E(\bO) \, .
\ee
Here $\epsilon_1$ and $E_1$ are the eigenvalue and projection,
respectively, fixed by the ground state $e_1$ of $h_1$. Making use of 
the estimate 
$\epsilon \mapsto (e^\epsilon - 1)^{-1} \leq e^{-k \epsilon}/(1-k) \epsilon$
for $\epsilon > 0$ and $0 < k < 1$, one obtains the upper bound,
putting $k = 1/2$, 
\be
\omega_{\beta, \epsilon_1, 1}(N_\omega(\bO))
\leq 2 (1/\beta(\epsilon_2 - \epsilon_1))^{-1} e^{(\beta/2) \, \epsilon_1}
  \, \text{Tr} \,
e^{- (\beta/2) h_1} \, .
\ee
The operators $e^{- (\beta/2) h_1}$ are of trace class, so it follows
again that, no matter how the trapping potential
is chosen, all particles with a wave function in
\mbox{$\cD = (1 - E_1) L^2(\bO)$} 
have a critical density for any
value of the temperature. As a matter of fact, since
the upper bound does not depend on the size of $\bO$, 
the total number of these particles in the limit state is finite. 
On the other hand, the  number of particles 
in~$\bO$ having a wave function which
overlaps with $e_1$ diverges if the chemical potential
approaches its maximal value. In more detail, for any
$f \in L^2(\bO)$ with $E_1f \neq 0$ one has,
cf.\ Proposition \ref{p.3.2}, 
\be \label{e.5.8}
\lim_{\mu \nearrow \epsilon_1} \,
\omega_{\beta, \mu, 1}((1 + a^*(f)a(f))^{-1}) = 0 \, .
\ee
Thus the density of particles in the ground state increases
unlimitedly in the approximating states, in contrast to
the density of all other particles in the thermal cloud escorting them.
So the particles in the ground state form Bose-Einstein condensates which can
be detected by observations in any given region $\bO$, as
outlined in the introduction. 

\medskip
In order to suppress the effects of the trapping potentials and    
to determine the inherent properties of the condensates,
it is convenient to proceed to the thermodynamic
limit. Let us briefly discuss this familiar topic from the
present point of view.
As was shown in Proposition \ref{p.3.2}, 
one obtains in the thermodynamic limit always the same
gauge invariant limit state $\omega_{\beta,\mu,\infty}$
for any given trapping potential
and $\beta > 0$, $\mu < 0$.  
The underlying scalar product is fixed by 
\be \label{e.5.9}
\langle f, f \rangle_{\beta,\mu,\infty} =
(1/2) \int \! d\bp \
{\frac{e^{\beta \, (\sbp^2 - \mu)} + 1}{e^{\beta \, (\sbp^2 - \mu)} -1}}
\ | \widetilde{f}(\bp) |^2 \, , \quad f \in L^2(\RR^s) \, ,
\ee
where the tilde $\, \widetilde{ } \,$ denotes Fourier transformation.

\medskip 
To probe for condensates in the
states $\omega_{\beta,\mu,\infty}$,    
we proceed again to the limit of maximal chemical
potential, \ie $\mu \nearrow 0$. 
In this limit the regular observables involve functions $f$ 
in the domain of~$|\bP|^{-1}$. Plugging these
functions into the preceding equation, the integrals
remain finite in the limit. Hence in terms of the regular
observables, the limit states admit a particle interpretation
in all regions~$\bO$. The question of whether the corresponding
particle numbers are summable, \ie whether the states are
locally normal, depends on the number of dimensions, however.

\medskip
In $s = 1,2$ dimensions the total    
number of particles with regular wave functions in~$\cD(\bO)$
tends to infinity in the limit, \ie the limit states
are not locally normal. This fact is in
general interpreted as absence of condensation.
But it should be noticed that for any given function
$f \in L^2(\bO)$ which does not lie in the domain of 
$|\bP|^{-1}$, the number of 
particles in $\bO$ with this particular wave function
tends also to infinity in the limit.
Thus it outnumbers the particles with wave functions
in any given finite dimensional subspace of the
regular space~$\cD(\bO)$. In this sense  
the (improper) states corresponding to the
spectral value $0$ of $| \bP |$, which are
orthogonal to the states in $\cD(\bO)$,  
form a quasi-condensate which manifests itself
in any given region $\bO$ in $s=1,2$ dimensions. 
  
\medskip 
In $s > 2$ dimensions all elements of $L^2(\bO)$    
are contained in the domain of~$| \bP |^{-1}$, in contrast
to the situation in lower dimensions. Moreover, the
\mbox{operators} 
\be
E(\bO) \, (e^{\beta \sbP^2} - 1)^{-1} E(\bO) \, ,
\quad \beta > 0 \, ,
\ee
being the product of a 
pair of adjoint Hilbert-Schmidt operators,
are of trace class. Thus the expectation values
of the (full) particle number operators $N(\bO)$
in the limit states are finite, so these states are
locally normal and possess a critical particle
density. It is homogeneous and given by
\be
\bx \mapsto \omega_{\beta,0,\infty}(a^*(\bx) a(\bx)) =
(2 \pi )^{-s} \! \int \! d\bp \, (e^{\beta \sbp^2} -1 )^{-1} \, .
\ee
In order to increase this density, one has to   
return to the approximating states and add to them
coherent configurations of (suitably renormalized) 
scaled ground states $e_{L,1}$ of the 
trapped Hamiltonian $h_L$. The resulting
states are normal with regard to the Fock representation.
Similarly to relation
\eqref{e.5.4}, one obtains by this procedure the
expectation values
\be \label{e.5.10} 
\omega_{\beta,\mu,L} \, \scirc \, \ad{W(e_{L,1})} \, \big(a^*(f) a(f) \big)
= \omega_{\beta,\mu,L} (a^*(f) a(f)) + |\langle e_{L,1}, f \rangle|^2 \, .
\ee
The local properties of the states appearing
in the thermodynamic and subsequent infinite particle number
limits can be read off from this relation by restricting
it to functions $f \in L^2(\bO)$. One can then replace
the functions $\bx \mapsto e_{L,1}(\bx) = e_1(\bx/L)$
by the pointwise products 
$\chi e_{L,1}$, where $\chi$ is the
characteristic function of $\bO$. Since 
ground state wave functions are continuous,
the functions $\chi e_{L,1}$ converge strongly
to $e_1(0) \chi$ in the limit of large $L$.
Proceeding first to the thermodynamic limit and then
to the limit $\mu \nearrow 0$, it follows from
equation \eqref{e.5.10} and the preceding
remarks that the resulting limit states
satisfy the KMS condition and are locally
normal. They describe the pure phases (primary states) 
in the central decomposition of 
condensed states, appearing in the thermodynamic limit
of Gibbs-von Neumann states in finite volume, cf.\ \cite[Thm.\ 5.2.32]{BrRo}
and \cite[Thm.\ III.3]{FaPuVe}. Their particle density is given~by
\begin{align}
\bx \mapsto \omega_{\beta,0,\infty}(a^*(\bx) a(\bx)) + |e_1(0)|^2 \, .
\end{align}
So the  density of the limit states is
enlarged by contributions from the (improper)
ground state $e_{\infty,1}$. These do not affect the
particle content of the original state and create a
condensate whose weight is fixed by the normalization
of~$e_1$. 

\medskip
We conclude this section with a remark which is of relevance in      
case of interacting systems. As already mentioned, one has to restrict
attention there to the field subalgebra $\fF \subset \fR$ in order to retain
control on the action of the dynamics. Since condensates 
in $s > 2$ dimensions typically appear as classical (mean) fields
shifting the quantum field, there arises the question of whether
the field algebra is stable under the adjoint action of Weyl operators,
\ie $\ad{W(e)}(\fF) \subset \fF$ for $e \in L^2(\RR^s)$. The answer
is affirmative, and we briefly sketch the argument for the simple case
of tensor fields which are constructed from a single resolvent.
So let $\Ree \lambda \neq 0$, let $f \in L^2(\RR^s)$
and let the norm of $e$ be sufficiently small such that
one obtains for the functional $l_e$ the upper bound
$\sup_u |l_e(e^{iu}f)|/|\Ree \lambda| < 1$. Clearly, 
$m \in \ZZ$, 
\be \label{e.5.13}
\ad{W(e)}\Big( \int_0^{2 \pi} \! du \, e^{imu} R(\lambda, e^{iu}f) \Big)
= \int_0^{2 \pi} \! du \, e^{imu}
R(\lambda - i \, l_e(e^{iu}f), e^{iu}f) \, ,
\ee
where the integrals are defined in the strong operator topology
on $\cF$. The resolvent on the right hand side of this equality
can be expanded in a Neumann series,
\be \label{e.5.14}
R(\lambda - i \, l_e(e^{iu}f), e^{iu}f) =
\sum_{n=1}^\infty (- l_e(e^{iu}f))^{n-1} \, R(\lambda,e^{iu}f)^n \, .
\ee
It is absolutely convergent in norm because of the limitations on 
$e$. Since the functional~$l_e$ is real linear,
one has $l_e(e^{iu}f) = \cos(u) \, l_e(f) + \sin(u) \, l_e(if)$.
It follows that each summand in the series \eqref{e.5.14}
gives rise in equation \eqref{e.5.13} 
to a finite sum of tensor operators in $\fF$
whose degree ranges between $m - n + 1$ and $m + n - 1$.
In view of the convergence properties of the series, this
shows that the operators \eqref{e.5.13} are contained in $\fF$
in this particular case. In a similar
manner one can treat tensor operators built by arbitrary
finite sums of products of resolvents. Iterating the adjoint
action of the Weyl operators, one finally sees that the constraints
on the norm of $e$ can be removed, establishing the statement.

\section{Conclusions}
\setcounter{equation}{0}

In the present article we have applied the framework of resolvent   
algebras in a study of thermal properties of trapped and
untrapped non-interacting bosons. In particular, we have
analyzed the properties of certain specific limit states, such as the
thermodynamic limit and the infinite particle number limit,
which are of substantial interest for the interpretation of the
theory. Compared to other
settings, this analysis is greatly simplified by the fact that
observables, which become singular in the
limit states, automatically disappear.
They are members of specific ideals which are annihilated
in the states. As a consequence, the limit states retain 
physically meaningful properties: they satisfy the KMS condition
on the full algebra and are regular on the subalgebra generated 
by the observables which remain non-trivial in the limit.

\medskip
The limit states lead to   
representations of the observable algebra which are disjoint
from the Fock representation, \ie they do not admit a particle
interpretation on the full algebra. It is therefore
natural to focus on local properties of the states which can be
determined in bounded regions of space. The resolvent algebra 
provides the necessary ingredients for such an analysis.
In case of the infinite particle number limit it
allowed us to determine local subalgebras of observables which
retain a meaningful particle interpretation in the limit states
and to calculate with their help the respective critical densities
of the thermal clouds. 
On the other hand the algebra contains local observables which
indicate the formation of condensates, outrunning these
critical densities. They lead to the 
breakdown of an associated particle picture in the limit states. 
This local point of view enabled us to treat in a
unified manner different manifestations of condensation
in trapped and untrapped equilibrium states. 

\medskip
Let us mention as an aside that the concept of    
localized observables plays an even more prominent role in the algebraic
approach to relativistic quantum physics~\cite{Ha}. It is based
on the insight that observations are always made in bounded
regions and that the physically accessible states 
are locally normal with respect to each other, 
in accordance with present results.   
The ensuing algebraic framework has led to numerous fundamental  
insights about the properties of the state space of relativistic
quantum field theories, its particle interpretation and 
sector structure,  the characterization 
of thermal states \etc. But it also led to new 
powerful algebraic schemes for the perturbative 
construction of models. Regarding the present issue 
of thermal bosons,  an intriguing
recent article by Brunetti, Fredenhagen and Pinamonti \cite{BrFrPi}
ought to be mentioned here, where states including Bose-Einstein
condensates are perturbatively constructed to all orders in an
interacting quantum field theory. 
See also the references quoted there for further
related results.

\medskip
We conclude this article with some remarks pertaining to the  
problem of condensation in trapped and untrapped interacting systems. 
To simplify the discussion,  we consider
two-body interaction potentials which are
continuous, vanish rapidly at infinity, and are repulsive.
For the trapped systems we restrict ourselves
to harmonic trapping potentials at length scale $L$.
The resulting Hamiltonians are well defined as 
selfadjoint operators $H_L$ on Fock space $\cF$. 
The adjoint action of the corresponding unitaries
$\ad{e^{itH_L}}$, $t \in \RR$, maps the algebra of
observables~$\fA$ into its canonical 
extension $\ofA$ and we also
consider the regular subalgebra $\obfA_c \subset \obfA$ on which it
acts pointwise norm continuously,  cf.~\cite{Bu3}.

\medskip
The equilibrium states of trapped systems for given  
$\beta > 0$ and $\mu < \mu_{\mbox{\tiny max}}$
(\ie the supremum of the admissible chemical potentials,   
appearing to be infinite by estimates based on the Bogolubov
approximation \cite{LiSeSoYn, Ve}) are
described on Fock space by the density operator 
$e^{-\beta(H_L - \mu N)}$, which is of trace class on $\cF$. Putting  
\be
\omega_{\beta, \mu, L}(A) \doteq
  Z^{-1} \, \mbox{Tr} \, e^{-\beta(H_L - \mu N)} \, A \, , \quad A \in \fA \, ,
\ee
where $Z$ denotes the partition function, one obtains Gibbs-von
Neumann states on the observable algebra $\fA$.
They extend by continuity to the regular algebra $\obfA_c$ and, 
since they are stationary, the expectation value of
any $A \in \bfA$ coincides with the expectation value
of its time-regularized version
$\int \! dt \, f(t) \, \ad{e^{itH_L}}(A) \in \obfA_c$
where $\int \! dt \, f(t) = 1$, cf.~equation \eqref{e.1.5}.

\medskip 
In order to check 
whether the bosons form condensates in these trapped states  
one must proceed to the supremum $\mu_{\mbox{\tiny max}}$
of the chemical potential and determine
the corresponding thermal cloud. It is noteworthy that
the limit points of the corresponding 
sequence of states $\mu \mapsto \omega_{\beta,\mu,L}$
on the regular algebra $\obfA_c$ still satisfy the KMS condition
at inverse temperature $\beta$ according to standard
results for C$^*$-dynamical systems \cite[Prop.\ 5.3.25]{BrRo}. 
Moreover, by the preceding regularization
procedure, the expectation
values of the resolvents of the particle number operators
$a^*(f)a(f)$, $f \in L^2(\RR^s)$, are defined in these 
limit states. So in order to determine the critical
density of the thermal cloud in a given region  
$\bO$ it is meaningful to compute the limits
\be
\lim_{\mu \nearrow \mu_{\mbox{\tiny max}}}
\omega_{\beta, \mu, L}(a^*(f) a(f)) \, ,
\quad f \in L^2(\bO) \, .
\ee
Functions $f$ for which this limit stays finite
form again a subspace $\cD(\bO) \subset L^2(\bO)$
corresponding to regular observables in the
limit state. In view of the repulsive nature of the 
interaction it seems plausible that this subspace
is non-trivial. It fixes the particle number
operator $N_\omega(\bO)$ of the thermal cloud, 
where the limit of $\mu \mapsto \omega_{\beta, \mu, L}(N_\omega(\bO))$
determines its critical density. The total number of particles
in $\bO$, determined by $N(\bO)$, is expected to grow
indefinitely in the trapped states in this limit.
Hence if the critical density
of the thermal cloud turns out to be finite, this would indicate  
the formation of unlimited amounts of condensate 
for increasing chemical
potential. The orthogonal complement of $\cD(\bO)$ in
$L^2(\bO)$ then describes the condensate in $\bO$. The co-dimension
of $\cD(\bO)$ may, however, be larger than one in the
presence of interaction due to the more complex structure of 
ground states.

\medskip
One may also discuss the issue of condensation in the thermodynamic limit
(no trapping potential), cf.\ \cite[Thm. 6.3.31]{BrRo}. 
There one knows from the outset that $\mu_{\mbox{\tiny max}} = 0$.
Moreover, because of the repulsive nature of the interaction, the
number of particles in bounded regions ought to be smaller than
in the non-interacting theory. We therefore conjecture that
in $s > 2$ dimensions the limit states for maximal chemical
potential are normal on the local field algebras  
$\fF(L^2(\bO))$ with regard to the
Fock representation,
\ie they have a critical local density. 
As was shown in the last part of the preceding section, one can increase
this density by the action of
Weyl automorphisms. Since the resulting states should be 
spatially homogeneous, one has to use  
Weyl operators depending on (improper) eigenstates of the
momentum operator with momentum $0$. One thereby arrives at states
which describe at given time a 
homogeneous, macroscopically occupied condensate accompanying the
thermal background. Due to the interaction between the
condensate and the background
these initial states are not stationary, however. The remaining
intricate problem is then to adjust the thermal background
states in a manner such that the composed states are in equilibrium again. 

\bigskip \noindent
\textbf{\large Acknowledgements} \\[1mm]  
We are grateful to Wojciech Dybalski and Daniel Weiss for fruitful suggestions
which helped us to improve the final version of this article.
DB would like to thank the Mathematics Institute of the University of
G\"ottingen for their generous hospitality.

\bigskip \noindent
\textbf{\large Data availability} \\[1mm]  
Data sharing is not applicable to this article as no new data were
created or analyzed in this study.

\end{document}